\documentclass[journal,onecolumn,12pt]{IEEEtran}
\usepackage{pslatex}
\usepackage{amsfonts,color,morefloats}
\usepackage{amssymb,amsmath,latexsym,amsthm}
\usepackage{cases}
\usepackage[noadjust]{cite}

\newcommand{\fq}{{\mathbb F}_{q}}

\newcommand{\C}{{\mathcal{C}}}

\newcommand{\mL}{{\mathcal{L}}}
\newcommand{\mP}{{\mathcal{P}}}

\newtheorem{thm}{Theorem}
\newtheorem{lem}{Lemma}
\newtheorem{cor}{Corollary}

\newtheorem{problem}{Problem}

\newtheorem{definition}{Definition}
\newtheorem{example}{Example}
\newtheorem{remark}{Remark}

\begin{document}
\title{On $(\mL,\mP)$-Twisted Generalized Reed-Solomon Codes}
\date{\today}

\author{Zhao Hu\thanks{Z. Hu and N. Li are with the Key Laboratory of Intelligent Sensing System and Security (Hubei University), Ministry of Education, the Hubei Provincial Engineering Research Center of Intelligent Connected Vehicle Network Security, School of Cyber Science and Technology, Hubei University, Wuhan, 430062, China. N. Li is also with the State Key Laboratory of Integrated Service Networks, Xi'an 710071, China. Email: zhao.hu@aliyun.com, nian.li@hubu.edu.cn},
Liang Wang\thanks{L. Wang and X. Zeng are with the Hubei Key Laboratory of Applied Mathematics, Faculty of Mathematics and Statistics, Hubei University, Wuhan, 430062, China. Email: liang.wang1@aliyun.com, xiangyongzeng@aliyun.com},
Nian Li, Xiangyong Zeng, Xiaohu Tang\thanks{X. Tang is with the Information Coding $\&$ Transmission Key Lab of Sichuan Province, CSNMT Int. Coop. Res. Centre (MoST), Southwest Jiaotong University, Chengdu, 610031, China. Email: xhutang@swjtu.edu.cn}}

\maketitle
\begin{abstract}
Twisted generalized Reed-Solomon (TGRS) codes are an extension of the generalized Reed-Solomon (GRS) codes by adding specific twists, which attract much attention recently. This paper presents an in-depth and comprehensive investigation of the TGRS codes for the most general form by using a universal method. At first, we propose a more precise definition to describe TGRS codes, namely $(\mL,\mP)$-TGRS codes, and provide a concise necessary and sufficient condition for $(\mL,\mP)$-TGRS codes to be MDS, which extends the related results in the previous works. Secondly, we explicitly characterize the parity check matrices of $(\mL,\mP)$-TGRS codes, and provide a sufficient condition for $(\mL,\mP)$-TGRS codes to be self-dual. Finally, we conduct an in-depth study into the non-GRS property of  $(\mL,\mP)$-TGRS codes via the Schur squares and the combinatorial techniques respectively. As a result, we obtain a large infinite families of non-GRS MDS codes.
\end{abstract}

\section{Introduction}
Let $q$ be a prime power, $\mathbb{F}_{q}$ denote the finite field with $q$ elements and $\mathbb{F}_{q}^*=\mathbb{F}_{q}\setminus\{0\}$. Let $m$ and $n$ be positive integers. Denote the $n$-dimensional vector space over $\mathbb{F}_{q}$ by $\mathbb{F}_{q}^n$ and the set of $m\times{n}$ matrices over $\mathbb{F}_{q}$ by $\mathbb{F}_{q}^{m\times{n}}$. An $[n,k,d]_q$ linear code $\C$ is a subspace of $\mathbb{F}_{q}^n$ with dimension $k$ and minimum Hamming distance $d$.
% For two vectors $\boldsymbol{x}=(x_1,...,x_n)$ and $\boldsymbol{y}=(y_1,...,y_n)$ in $\mathbb{F}_q^n$, the inner product of $\boldsymbol{x}$ and $\boldsymbol{y}$ is given by $\boldsymbol{x}\cdot\boldsymbol{y}=\sum_{i=1}^nx_iy_i$.
The dual code $\C^\bot$ of a linear code $\C$ is defined by
\begin{eqnarray*}
\C^\bot=\{\boldsymbol{x} \in \fq^n \mid\boldsymbol{x}\cdot\boldsymbol{y}=0 \mbox{ for all } \boldsymbol{y}\in \C\},
\end{eqnarray*}
where $x\cdot y$ denotes the Euclidean inner product of $x$ and $y$.
A code $\C$ is called self-dual if $\C=\C^{\perp}$. Self-dual codes are an important class in algebraic coding theory since their important applications in secret sharing schemes, quantum communication, and error correction capability optimization \cite{CRAM,DOUG,STEA,FORN}.

For an $[n,k,d]$ linear code $\C$, the Singleton bound \cite{SING} implies that $d\leq n-k+1$, and the Singleton defect of  $\C$ is defined by $S(\C)=n-k+1-d$ \cite{AMAR}. If $S(\C)=0$, the code $C$ is called a maximum distance separable (MDS) code. If $S(\C)=1$, the code is called an almost-MDS (AMDS) code. If $S(\C)=S(\C^{\perp})=1$, then $\C$ is referred to as a near-MDS (NMDS) code. More generally, a code is called an $m$-MDS code if $S(\C)=S(\C^{\perp})=m$.
MDS codes are highly valued in information storage due to their optimal trade-off between storage capacity and reliability. Given that MDS and NMDS codes play an essential role in coding theory and have a wide range of applications, the study of these codes has attracted significant attention, involving their classification, construction, self-duality and inequivalence; see, for example, \cite{AMAR,BART,BEBO,BEPU,BERO,BETS,CASC,CHEN,COUV,CRAM,DING,DODU,GU,GULL,HUAQ,HUAN,LAVA,LIUH,ROTL,ROTS,SINC,SUIL,SUIS,SUIZ,ZHANGJ,ZHANGA,ZHAO}. The best known MDS codes are the so-called Reed-Solomon (RS) codes, which have significant applications such as in cryptography and distributed storage systems. Moreover, the construction of self-dual MDS codes from GRS codes has been extensively studied, and some related works are summarized in \cite{ZHANGA}.

TGRS codes are an extension of GRS codes, which was originally initiated by Beelen et al. \cite{BEPU} in 2017. Unlike GRS codes, TGRS codes are not necessarily MDS codes. Accordingly, constructing MDS codes from TGRS codes by adding different twists attracts much attention from researchers. Moreover, it is shown that TGRS codes have good structure properties
which making that TGRS codes can be applied as a promising alternative to Goppa codes in the McEliece code-based cryptosystem \cite{BEBO}. 
Due to the efficiency of constructing MDS codes from TGRS codes and their potential in cryptographic applications, TGRS codes have garnered significant attention in recent research.

By adding certain monomials (referred to as twists) to specific positions (referred to as hooks) of each generating polynomial $f(x)$ of GRS codes, TGRS codes can be obtained from GRS codes. We refer to it as the $(\mL,\mP)$-TGRS code in this paper, where $\mL$ (resp. $\mP$) denotes the twist set (resp. position set), see Definition \ref{D2} for more details.
%To the best of our knowledge, in 2017, Beelen et al. \cite{BEPU} first introduced the concept of TGRS codes. At the beginning, $1$-TGRS codes (i.e. $(\mL,\mP)$-TGRS codes with $\ell:=|\mL|=1$) attracts researchers' attention. 
In the initial stage, $1$-TGRS codes (i.e. $(\mL,\mP)$-TGRS codes with $\ell:=|\mL|=1$) attracts the interest of researchers.
Let $\mL=\{t\}$ and $\mP=\{h\}$. In 2017, Beelen et al. \cite{BEPU} characterized the necessary and sufficient condition for $1$-TGRS codes to be MDS, and presented two families of MDS $1$-TGRS codes for the cases that $(t,h)=(1,0)$ and $(t,h)=(1,k-1)$. Later, for the case $(t,h)=(1,k-1)$, Huang et al. \cite{HUAN} determined the parity check matrices of $1$-TGRS codes, and presented a necessary and sufficient condition such that $1$-TGRS codes are self-dual. Zhang et al. \cite{ZHANGJ} explored the minimum distance and dual codes of $1$-TGRS codes for $(t,h)=(q-k-1,0\leq h \leq{k-1})$. Furthermore, for any pair $(t,h)$, Sui et al. \cite{SUIZ} provided necessary and sufficient conditions for $1$-TGRS codes to be MDS and NMDS respectively.

% For $\ell=1$, let $L=\{t\}$, $I=\{h\}$ and $S=\{b:\ b\in{\mathbb{F}_{q}^*}\}$, we obtain the $1$$\mbox{-}$TGRS code (See Remark \ref{R1}) studied in \cite{BEPU}.
% In that paper, Beelen et al. first introduced the definition of a $1$$\mbox{-}$TGRS code and characterized the necessary and sufficient conditions for it to be a MDS code. They also identified two important subclasses of MDS $1$$\mbox{-}$TGRS codes for the cases where $(t,h)=(1,0)$ and $(t,h)=(1,k-1)$. Following this, Huang et al. conducted an in-depth study of the $1$$\mbox{-}$TGRS code specifically for the case $(t,h)=(1,k-1)$ in \cite{HUAN}. They not only provided the parity-check matrix for the $1$$\mbox{-}$TGRS code, but also derived the necessary and sufficient conditions for it to be self-dual. In another significant contribution, Zhang et al. explored the minimum distance and dual code of the $1$$\mbox{-}$TGRS code for the case where $(t,h)=(q-k-1,\ell)$, where $0\leq{\ell}\leq{k-1}$ \cite{ZHANGJ}. Additionally, Sui et al. characterized the necessary and sufficient conditions for a $1$$\mbox{-}$TGRS code to be both MDS and NMDS for any given position pair $(t,h)$ \cite{SUIZ}.

After that, scholars are dedicated to studying $(\mL,\mP)$-TGRS codes with $\ell>1$. Beelen et al. \cite{BEBO} first proposed a general form of $(\mL,\mP)$-TGRS codes with $\ell=|\mP|\leq\min\{k,n-k\}$ for a special coefficient matrix $B$,
%(where $\mL=\{t_1,t_2,...,t_{\ell}\}$ and $\mP=\{h_1,h_2,...,h_{\ell}\}$ and the coefficient matrix $B$ satisfy $b_{h_i,t_j}=0\ \mbox{for}\ i\neq j$ and $b_{h_i,t_j}\neq0$ , otherwise)},
and they constructed some MDS codes. Furthermore, in 2022, Beelen et al. \cite{BERO} take an in-depth discussion on the inequivalence of such TGRS codes to GRS codes and the decoding algorithm of these codes, and constructed infinite families of non-GRS MDS codes. Subsequently, some new results on the TGRS codes proposed by \cite{BEBO} were presented. Based on this form of TGRS codes, Sui et al. \cite{SUIL} focused on $2$-TGRS codes with $\mL=\{0,1\}$ and $\mP=\{k-1,k-2\}$, in which necessary and sufficient conditions for such $2$-TGRS codes to be MDS and self-dual were presented respectively, and infinite families of MDS (resp. NMDS, $2$-MDS) self-dual TGRS codes were obtained. Moreover, they demonstrated that most of their $2$-TGRS codes are non-GRS. Later, Gu et al. \cite{GU} constructed infinite families of self-dual MDS codes from the $(\mL,\mP)$-TGRS codes proposed by \cite{BEBO}, where $\ell < \min\{k,n-k\}$, $\mL=\{0,1,...,\ell-1\}$ and $\mP=\{k-\ell,k-\ell+1,...,k-1\}$. Harshdeep et al. \cite{SINC} provided a necessary and sufficient condition for the $(\mL,\mP)$-TGRS codes proposed by \cite{BEBO} to be MDS. Furthermore, Cheng \cite{CHEN} gave an explicit expression for the parity check matrices of TGRS codes of this form. 

Note that the $(\mL,\mP)$-TGRS codes with $\mL=\{0,1,...,n-k-1\}$ and $\mP=\{0,1,...,k-1\}$ is the most general case of TGRS codes. In 2023, Sui et al. \cite{SUIS} proposed this form of TGRS codes, and for the case when $\mL=\{0,1\}$ and $\mP=\{k-2,k-1\}$ they characterized the necessary and sufficient conditions for $(\mL,\mP)$-TGRS codes to be MDS and self-dual respectively. In 2024, Ding et al. \cite{DING} presented the necessary and sufficient conditions for $(\mL,\mP)$-TGRS codes with general coefficient matrix $B$ to be MDS and self-dual respectively, where $\ell<\min\{k,n-k\}$, $\mL=\{0,1,...,\ell-1\}$ and $\mP=\{k-\ell,k-\ell+1,...,k-1\}$. Recently, Zhao et al. \cite{ZHAO} provided a necessary and sufficient condition for the general $(\mL,\mP)$-TGRS codes to be MDS.

In summary, there exist many nice works dedicated to studying $(\mL,\mP)$-TGRS codes for specific twist set $\mL$ and position set $\mP$, including the construction of MDS codes, NMDS codes and self-dual codes from the TGRS codes and the characterization of the parity check matrices and equivalence of these TGRS codes. However, for the most general case where $\mP=\{0,1,...,k-1\}$ and $\mL=\{0,1,...,n-k-1\}$, there is a lack of more comprehensive and in-depth research on the $(\mL,\mP)$-TGRS codes.

% {\bf Question:} Can we investigate the $(\mL,\mP)$-TGRS codes of the most general form (namely, $\mP=\{0,1,...,k-1\}$ and $\mL=\{0,1,...,n-k-1\}$) comprehensively by using a universal method?

In this paper, we mainly investigate the TGRS codes for the most general case. At first, we introduce a more generic and precise definition for the TGRS codes, namely, the $(\mL,\mP)$-TGRS codes (see Definition \ref{D2}). Secondly, we present a necessary and sufficient condition such that the $(\mL,\mP)$-TGRS codes for the most general case are MDS, which is more concise and simpler than that in \cite{ZHAO} and extends some related results in the previous works. We also provide a necessary and sufficient condition for $(\mL,\mP)$-TGRS codes to be NMDS under the assumption that it is self-dual. Thirdly, we accurately characterize the parity check matrices of $(\mL,\mP)$-TGRS codes and propose a sufficient condition for $(\mL,\mP)$-TGRS codes to be self-dual. Finally, we study the non-GRS properties of $(\mL,\mP)$-TGRS codes by using the Schur squares and combinatorial techniques respectively. As a result, we obtain an infinite families of non-GRS MDS codes.

% In the second part of this paper, we present three lemmas. Using these lemmas, we can directly compute the determinant of TGRS codes with arbitrary twists and obtain the necessary and sufficient conditions for these codes to be MDS. The necessary and sufficient conditions provided in this paper are more intuitive and simpler than those in \cite{ZHAO}. Subsequently, we demonstrate through five corollaries that when the coefficient matrix takes special forms, our results can cover the MDS results for TGRS codes presented in previous literature.

% We then briefly investigate the necessary and sufficient conditions for TGRS codes to be NMDS codes. In the third part of this paper, we compute the parity-check matrix for TGRS codes with arbitrary twists. From the parity-check matrix, the dual code of the TGRS code can be obtained. We then study the sufficient conditions for the TGRS code to be self-dual. In the final section of the paper, we use Schur squares and impose certain restrictions on the coefficient matrix to find a new infinite class of TGRS codes with rate less than 1/2. This infinite class contains many MDS codes that are not equivalent to GRS codes. By analogy with \cite{BERO}, we further prove the non-GRS nature of TGRS codes with arbitrary twists using combinatorial inequivalence. Finally, we conclude the paper with a summary of our results.

This paper is organized as follows. In Section \ref{sect-2}, we introduce some notation, definitions and fundamental results with respect to $(\mL,\mP)$-TGRS codes. In Section \ref{sect-3}, we determine a necessary and sufficient condition for $(\mL,\mP)$-TGRS codes to be MDS and a sufficient condition for these codes to be NMDS. In Section \ref{sect-4}, we characterize the parity check matrices and dual codes of $(\mL,\mP)$-TGRS codes. In Section \ref{sect-5}, we investigate the non-GRS properties of $(\mL,\mP)$-TGRS codes, and obtain a large family of non-GRS MDS codes. Section \ref{sect-6} concludes this paper.

\section{Preliminaries} \label{sect-2}
In this section, we introduce some notation, definitions and lemmas which will be used in subsequent sections.
Starting from now on, we adopt the following notation unless otherwise stated:
\begin{itemize}
  \item Let $q$ be a  prime power, and $m$, $n$ and $k$ be positive integers with $k\leq{n}$.
  \item Let $\mathbb{F}_{q}$ denote the finite field of order $q$ and $\mathbb{F}_{q}^*=\mathbb{F}_{q}\backslash\{0\}$.
  \item Let $\mathbb{F}_{q}^n$ denote the $n$-dimensional vector space over $\mathbb{F}_{q}$ and 
  $\mathbb{F}_{q}^{m\times{n}}$ denote the set of $m\times{n}$  matrices over $\mathbb{F}_{q}$.
  \item Let $[k]:=\{0,1,\ldots,k\}$ and $|S|$ denote the cardinality of a set $S$.
  \item Let $A^{\rm T}$ denote the transpose of a matrix $A$.
  % \item Let $\mathbb{N}$ denote the set of nonnegative integers and $\ell\in \mathbb{N}$.
  % \item Let $\mathbb{F}_{q}[x]$ be the polynomial ring over $\mathbb{F}_{q}$ and $\mathbb{F}_{q}[x_1\cdots x_m]$ denote the multivariate polynomial ring over $\mathbb{F}_{q}$.
  % \item Let $\mathbb{F}_{q}[x]_{<k}$ denote the set of $f(x)\in{\mathbb{F}_{q}[x]}$ with $\deg(f(x))<k$, and similarly $\mathbb{F}_{q}[x]_{<n}$ for polynomials with degree less than $n$.
  \item For $\boldsymbol{\alpha}=(a_1,\ldots,a_n)\in\mathbb{F}_{q}^n$ and $\boldsymbol{\nu}=(v_1,\ldots,v_n)\in(\mathbb{F}_{q}^*)^n$, let $ev_{\boldsymbol{\alpha},\boldsymbol{\nu}}$ denote an evaluation map from $\mathbb{F}_{q}[x]$ to $\mathbb{F}_{q}^n$ with $ev_{\boldsymbol{\alpha},\boldsymbol{v}}(f(x))=(v_1f(a_1),\ldots,v_nf(a_n))$.
  %Here, $\alpha$ are called evaluation points.
  \item Let $\boldsymbol{x}\ast\boldsymbol{y}=(x_1y_1,\ldots,x_ny_n)$ denote the componentwise product of $\boldsymbol{x}$ and $\boldsymbol{y}$, where $\boldsymbol{x}=(x_1,\ldots,x_n), \boldsymbol{y}=(y_1,\ldots,y_n) \in \mathbb{F}_q^n$.
  % \item For a matrix $B\in{\mathbb{F}_{q}^{n\times{n}}}$, $|B|$ denotes the determinant of $B$.
  % \item $I_k$ denotes the $k\times{k}$ identity matrix.
\end{itemize}
% {\color{blue}It is easy to prove that the map $ev_{\boldsymbol{a},\boldsymbol{v}}$ satisfies both the closure under addition and the closure under scalar multiplication, which implies that $ev_{\boldsymbol{a},\boldsymbol{v}}$ is a linear map.} {\color{red} hint: is it necessary?}

\subsection{GRS codes}
GRS codes are a well-known family of MDS codes, which are generally very useful in many applications \cite{HWPV}. We recall the definition of GRS codes as follows.  
\begin{definition}(\cite{HWPV})\label{D1}
Let $\boldsymbol{\alpha}=(a_1,a_2,...,a_n)\in\mathbb{F}_{q}^n$ with $a_1,...,a_n$ distinct, $\boldsymbol{\nu}=(v_1,...,v_n)\in(\mathbb{F}_{q}^*)^n$, where $n$ and $k$ are  positive integers such that $0\leq{k}<n$. The generalized Reed-Solomon (GRS) code $\C(\boldsymbol{\alpha},\boldsymbol{\nu})$ is defined as
\begin{equation}\label{S2}
\C(\boldsymbol{\alpha},\boldsymbol{\nu})=\{ev_{\boldsymbol{\alpha},\boldsymbol{\nu}}(f(x))=(v_1f(a_1),v_2f(a_2),...,v_nf(a_n)):\ f(x)\in{\mathbb{F}_{q}[x]_{<k}}\},
\end{equation}
where $\mathbb{F}_{q}[x]_{<k}=\{\sum_{i=0}^{k-1}f_ix^i: f_i \in \fq,\,0\leq i \leq k-1 \}$ denotes the set of polynomials $f(x)\in{\mathbb{F}_{q}[x]}$ with $\deg(f(x))<k$.
When $\nu=(1,\ldots,1)$, this code is referred to as the Reed-Solomon (RS) code.
\end{definition}

Since $ev_{\boldsymbol{\alpha},\boldsymbol{v}}$ is a linear map and the set $\mathbb{F}_{q}[x]_{<k}$ forms a vector space of dimension $k$, the code $\C(\boldsymbol{\alpha},\boldsymbol{\nu})$ is a linear subspace of $\mathbb{F}_q^n$. The code $\C(\boldsymbol{\alpha},\boldsymbol{\nu})$ is an $[n,k,n-k+1]$ MDS code over $\fq$ \cite{HWPV}.
% {\color{blue}The dimension of the image equals the dimension of the domain of $ev_{\boldsymbol{a},\boldsymbol{v}}$, which is $k$, because the map is linear and injective (since the evaluation points $\boldsymbol{\alpha}$ are distinct, ensuring that distinct polynomials map to distinct codewords). Therefore $\C_{n,k}(\boldsymbol{\alpha},\boldsymbol{\nu})$ is a $k$$\mbox{-}$dimensional linear code. Since the polynomial $f(x)$ can have at most $k-1$ roots in $\mathbb{F}_{q}$, the code $\C_{n,k}(\boldsymbol{\alpha},\boldsymbol{\nu})$ achieves the maximum possible minimum distance, making it a Maximum Distance Separable (MDS) code \cite{WCHU}.} {\color{red} delete something.}
Moreover, it is known that the generator matrix of the GRS code $\C(\boldsymbol{\alpha},\boldsymbol{\nu})$ is given by \begin{equation}\label{S3}
G_{GRS}=\begin{pmatrix}
v_1&v_2&\cdots&v_n\\
v_1a_1&v_2a_2&\cdots&v_na_n\\
\vdots&\vdots&\ddots&\vdots\\
v_1a_1^{k-1}&v_2a_2^{k-1}&\cdots&v_na_n^{k-1}\\
\end{pmatrix}.
\end{equation}

\subsection{$(\mL,\mP)$-TGRS codes}

TGRS codes are an extension of GRS codes by adding certain monomials (referred to as twists) to specific positions (referred to as hooks) of each polynomial $f(x)=\sum_{i=0}^{k-1}f_ix^i$ of GRS codes, where $f_i\in \fq$ for $0\leq i \leq k-1$.
Although the definition of TGRS codes have been given in the previous works \cite{BEPU,HUAN,ZHANGJ,SUIZ,BEBO,CHEN,GU,SINC,SUIL,SUIS,DING,ZHAO} in different ways (maybe for certain monomials and positions), we provide a unified definition in the following.

\begin{definition}\label{D2}
Let $n,k$ and $\ell$ be integers with $0<k\leq n$ and $0\leq \ell \leq n-k$. Choose the following three notation:
\begin{itemize}
  \item $\mathcal{L}\subseteq [n-k-1]$ (called the twist set), where $\ell:=|\mathcal{L}|$ denotes the number of twists;
  \item $\mathcal{P}\subseteq[k-1]$ (called the position set);
  \item $B=[b_{i,j}] \in \fq^{k\times{(n-k)}}$ (called the coefficient matrix), where $0\leq i \leq k$ and $0\leq j \leq n-k-1$.
\end{itemize}
% Let $\mathcal{L}\subseteq [n-k-1]$ be the twist set with $\ell=|\mathcal{L}|$, and $\mathcal{I}\subseteq[k-1]$ be the position set.
% We also define a $k\times{(n-k)}$ coefficient matrix $B=[b_{i,j}]$ over $\fq$, where $b_{i,j}$'s are entries of $B$.
For given $\mL$, $\mP$ and $B$, the set of twisted polynomials is given by
\begin{equation}\label{S5}
F(\mL,\mP,B)=\left\{\sum_{i=0}^{k-1}f_ix^i+\sum_{i\in \mP}f_i\sum_{j\in \mL} b_{i,j}x^{k+j}: f_i \in \fq,\, 0\leq i \leq k-1\right\}.
\end{equation}
Let $\boldsymbol{\alpha}=(a_1,a_2,...,a_n)\in\mathbb{F}_{q}^n$ with distinct $a_1,...,a_n$, and ${\bf \nu}=(v_1,...,v_n)\in(\fq^*)^n$. Then the twisted generalized Reed-Solomon (TGRS) code is defined by
\begin{equation}\label{S6}
\C(\mL,\mP,B)=\{ev_{\boldsymbol{\alpha},\boldsymbol{v}}(f(x))=(v_1f(a_1),...,v_nf(a_n)):\ f(x)\in{F_{n,k}(\mL,\mP,B)}\}.
\end{equation}
For simplicity and accuracy, we call it $(\mL,\mP)$-TGRS code throughout this paper. It is also called $\ell$-TGRS in other literature.
Specifically, when ${\bf\nu}=(1,\ldots,1)$, the code is referred to as $(\mL,\mP)$-TRS code.
% In the special case where $\boldsymbol{\nu}=\textbf{1}=(1,...,1)$, the code
% \begin{equation}\label{S7}
% \C_{n,k}(\boldsymbol{\alpha},\textbf{1},I,L,S)=\{ev_{\boldsymbol{a}}(f(x))=(f(a_1),...,f(a_n)):\ f(x)\in{F_{n,k}(I,L,S)}\}
% \end{equation}
% is referred to as the twisted Reed-Solomon ($\textbf{TRS}$) code.
\end{definition}

Accordingly, the coefficient matrix $B$ of $(\mL,\mP)$-TGRS codes is given by
\begin{equation}\label{S8-B}
  B=\begin{pmatrix}
    b_{0,0}&b_{0,1}&\dots&b_{0,n-k-1}\\
    b_{1,0}&b_{1,1}&\dots&b_{1,n-k-1}\\
    \vdots&\vdots&\ddots&\vdots\\
    b_{k-1,0}&b_{k-1,1}&\dots&b_{k-1,n-k-1}\\
  \end{pmatrix},
\end{equation}
where $b_{i,j}\in \fq$ for $0\leq i\leq k-1$ and $0\leq j \leq n-k-1$.

By selecting proper $\mathcal{L}\subseteq [n-k-1]$, $\mathcal{P}\subseteq[k-1]$ and $B$, the $(\mL,\mP)$-TGRS codes will be reduced to the TGRS codes defined in the previous works. Note that all $(\mL,\mP)$-TGRS codes for any $\mathcal{L}\subseteq [n-k-1]$ and $\mathcal{P}\subseteq [k-1]$ can be obtained from $(\mL,\mP)$-TGRS codes with $\mathcal{L}=[n-k-1]$ and $\mathcal{P}=[k-1]$ by setting proper coefficient matrix $B$, since if the $i$-th row and $j$-th column of $B$ for $i\in \mP'\subseteq [k-1]$ and $j\in\mL'\subseteq [n-k-1]$ are all zero then the $(\mL,\mP)$-TGRS codes with $\mathcal{L}=[n-k-1]$ and $\mathcal{P}=[k-1]$ are reduced to the $(\mL,\mP)$-TGRS codes with $\mathcal{L}=[n-k-1]\setminus \mL'$ and $\mathcal{P}=[k-1]\setminus \mP'$.
Clearly, $(\mL,\mP)$-TGRS codes with $\mathcal{L}=[n-k-1]$ and $\mathcal{P}=[k-1]$ are the most general case for TGRS codes. Throughout this paper, we always focus on the most general case and assume that $\mathcal{L}=[n-k-1]$ and $\mathcal{P}=[k-1]$.

Next, we will explore the properties and generator matrices of $(\mL,\mP)$-TGRS codes.

\begin{lem}\label{L1}
Let $\boldsymbol{\alpha}=(a_1,a_2,...,a_n)\in\mathbb{F}_{q}^n$ with distinct $a_1,...,a_n$, $\boldsymbol{\nu}=(v_1,...,v_n)\in(\mathbb{F}_{q}^*)^n$ and $B$ be defined as in \eqref{S8-B}. Then we have the following:
\begin{enumerate}
  \item The set of twisted polynomials $F(\mL,\mP,B)$ defined as in \eqref{S5} is a $k$-dimensional subspace of $\mathbb{F}_q[x]$, and the set $\{g_i(x): 0 \leq i \leq k-1\}$ is a basis of $F(\mL,\mP,B)$, where
      \begin{equation}\label{S7}
      g_i(x)=x^i+\sum_{j=0}^{n-k-1}b_{i,j}x^{k+j},\ 0\leq{i}\leq{k-1}.
      \end{equation}
  \item The $(\mL,\mP)$-TGRS code $\C(\mL,\mP,B)$ defined as in \eqref{S6} is an $[n,k]$ linear code with the generator matrix
        \begin{equation}\label{S8-GM}
        G_{TGRS}=\begin{pmatrix}
        ev_{\boldsymbol{\alpha},\boldsymbol{\nu}}(g_0(x))\\
        \vdots\\
        ev_{\boldsymbol{\alpha},\boldsymbol{\nu}}(g_{k-1}(x))\\
        \end{pmatrix}
        =\begin{pmatrix}
        v_1(1+\sum\limits_{j=0}^{n-k-1}b_{0,j}a_1^{k+j})&\cdots&v_n(1+\sum\limits_{j=0}^{n-k-1}b_{0,j}a_n^{k+j})\\
        v_1(a_1+\sum\limits_{j=0}^{n-k-1}b_{1,j}a_1^{k+j})&\cdots&v_n(a_n+\sum\limits_{j=0}^{n-k-1}b_{1,j}a_n^{k+j})\\
        \vdots&\ddots&\vdots\\
        v_1(a_1^{k-1}+\sum\limits_{j=0}^{n-k-1}b_{k-1,j}a_1^{k+j})&\cdots&v_n(a_n^{k-1}+\sum\limits_{j=0}^{n-k-1}b_{k-1,j}a_n^{k+j})
        \end{pmatrix}.
        \end{equation}
  \end{enumerate}
\end{lem}
\begin{proof}
Let $f(x)=\sum_{i=0}^{k-1}f_ix^i+\sum_{i=0}^{k-1}f_i\sum_{j=0}^{n-k-1}b_{i,j}x^{k+j}\in F(\mL,\mP,B)$, where $f_i\in{\mathbb{F}_q}$, and $g_i(x)$'s are given as in \eqref{S7} for $0\leq{i}\leq{k-1}$.
We claim that $g_0(x),g_1(x),...,g_{k-1}(x)$ are linearly independent over $\fq$. Observe that each $g_i(x)$ contains a distinct monomial $x^i$, where $0\leq{i}\leq{k-1}$, and $\{x^i: 0\leq{i}\leq{k-1}\}$ is absolutely a basis of $\mathbb{F}_{q}[x]_{<k}$.
Additionally, the degree of $\sum_{j=0}^{n-k-1}b_{i,j}x^{k+j}$ in $g_{i}(x)$ is at least $k$. This means that $g_0(x),g_1(x),...,g_{k-1}(x)$ are linearly independent over $\fq$.
Since $|F(\mL,\mP,B)|=q^k$ by the definition and it is closed under the addition and scalar multiplication, $\{g_i(x): 0 \leq i \leq k-1\}$ is a basis of the vector space $F(\mL,\mP,B)$, namely, each $f(x)$ can be expressed as $f(x)=\sum_{i=0}^{k-1}f_ig_i(x)$.
This proves 1).

Note that $ev_{\boldsymbol{\alpha},\boldsymbol{\nu}}$ is a linear one-to-one mapping. Thus the $(\mL,\mP)$-TGRS code $\C(\mL,\mP,B)$ is an $[n,k]$ linear code and $\{ev_{\boldsymbol{\alpha},\boldsymbol{\nu}}(g_0(x)),...,ev_{\boldsymbol{\alpha},\boldsymbol{\nu}}(g_{(k-1)}(x))\}$ is a basis of $\C(\mL,\mP,B)$, which gives the generator matrix $G_{TGRS}$ of $\C(\mL,\mP,B)$. This completes the proof.
\end{proof}

Moreover, it should be noted that the generator matrix of the $(\mL,\mP)$-TGRS code can be expressed as
\begin{equation}\label{S9-G}
G_{TGRS}=[I_k|B]V_nV_0,
\end{equation}
where $I_k$ is the $k\times{k}$ identity matrix, $V_n$ is the $n\times{n}$ Vandermonde matrix and $V_0$ is a diagonal matrix with elements $\{v_1,v_2,...,v_n\}$, given by
\begin{equation*}
V_n=\begin{pmatrix}
1&1&\cdots&1\\
a_1&a_2&\cdots&a_n\\
\vdots&\vdots&\ddots&\vdots\\
a_1^{n-1}&a_2^{n-1}&\cdots&a_n^{n-1}\\
\end{pmatrix},
V_0=\begin{pmatrix}
v_1&&&\\
&v_2&&\\
&&\ddots&\\
&&&v_n
\end{pmatrix}.
\end{equation*}

\subsection{Equivalence of linear codes}

In the following, we introduce the equivalence of linear codes over $\fq$.

\begin{definition}(\cite{BEPU}) \label{D3}
Let $\C_1$ and $\C_2$ be linear codes over $\fq$ with length $n$. We say that $\C_1$ and $\C_2$ are equivalent if there is a permutation $\pi$ in the permutation group with order $n$ and $\boldsymbol{\nu}=(v_1,...,v_n)\in{(\mathbb{F}_q^*)^n}$ such that
\begin{equation*}
\C_2=\Phi_{\pi,\boldsymbol{\nu}}(\C_1),
\end{equation*}
where $\Phi_{\pi,\boldsymbol{\nu}}: \mathbb{F}_q^n \rightarrow \mathbb{F}_q^n$ is defined by
\begin{equation*}
(c_1,\ldots,c_n)\mapsto(v_1c_{\pi(1)},\ldots,v_nc_{\pi(n)}).
\end{equation*}
\end{definition}

The equivalence preserves essential properties of a linear code, including the length, minimum distance, dimension, generator and parity check matrices, dual code, and automorphism group \cite{HWPV}.

% These properties remain unchanged because the equivalence involves permutations and scalar multiplications, which do not affect the fundamental structure of the code, such as the linear independence of codewords, their distances, or the code's generator and dual relationships.

\begin{remark}\label{R2}
Accordingly, the $(\mL,\mP)$-TGRS code $\C(\mL,\mP,B)$ for any ${\bf \nu} \in(\fq^*)^n$ is equivalent to $\C(\mL,\mP,B)$ with ${\bf\nu}=(1,...,1)$.
% equivalent to $\C_{n,k}(\boldsymbol{\alpha},\boldsymbol{\nu},I,L,S)$ and $\C_{n,k}(\boldsymbol{\alpha},\textbf{1},I,L,S)$ are equivalent.
% Clearly, $\C(\boldsymbol{\alpha},\boldsymbol{\nu},B)$ and $\C(\boldsymbol{\alpha},\textbf{1},B)$ are also equivalent.
\end{remark}

\section{MDS $(\mL,\mP)$-TGRS codes} \label{sect-3}
% {\color{blue}A Maximum Distance Separable (MDS) code is a linear code that achieves the maximum possible minimum Hamming distance between any pair of codewords for a given code length and code dimension. Formally, a linear code $\C$ of length $n$ and dimension $k$ over a finite field $\mathbb{F}_q$ is called MDS if the minimum Hamming distance $d$ satisfies $d=n-k+1$.}

The study of MDS codes is of great significance because they provide optimal error detection and correction capabilities. This makes them indispensable in areas such as communication, data storage, and coding theory. MDS codes form an essential family of codes in coding theory.

In this section, we will investigate the MDS properties of $(\mL,\mP)$-TGRS codes for the most general case.  We first show some useful lemmas.

\begin{lem}(\cite{HWPV})\label{L2}
Let $\C$ be an $[n,k]$ linear code over $\mathbb{F}_q$. Let $G$ be a generator matrix of $\C$. Then $\C$ is an MDS code if and only if every $k\times{k}$ minor (determinant of a $k\times{k}$ submatrix) of $G$ is nonzero.
\end{lem}

\begin{lem}({\cite[Lemma III.1]{GU}})\label{L3}
Let $A_t$ be a $(t+1)\times (t+1)$ matrix over $\fq$ given by
$$A_t=\begin{pmatrix}
c_0&&&&\\
c_1&c_0&&&\\
c_2&c_1&c_0&&\\
\vdots&\vdots&\ddots&\ddots\\
c_{t}&c_{t-1}&\cdots&c_1&c_0
\end{pmatrix}, $$
where $c_0=1$ and $c_1,c_2,...,c_t\in{\mathbb{F}_{q}}$ for a nonnegative integer $t$. Then the inverse of $A_t$ is
\begin{equation*}
A_t^{-1}=\begin{pmatrix}
e_0&&&&\\
e_1&e_0&&&\\
e_2&e_1&e_0&&\\
\vdots&\vdots&\ddots&\ddots\\
e_{t}&e_{t-1}&\cdots&e_1&e_0
\end{pmatrix},
\end{equation*}
where $e_0=1$ and $e_i=-\sum_{j=0}^{i-1}e_jc_{i-j}$ for $0\leq i \leq t$.
\end{lem}

\begin{lem}\label{L4}
Let $\boldsymbol{\alpha}=(a_1,a_2,...,a_n)\in\mathbb{F}_{q}^n$ with distinct $a_i$'s, and $\mathcal{T}=\{t_1,\ldots,t_k\}$ be a $k$-subset of $\{1,...,n\}$.
Let $\prod_{i=1}^{k}(x-a_{t_i})=\sum_{j=0}^{k}c_jx^{k-j}$, where $c_j$'s are uniquely determined by $a_{t_i}$'s. For  any $0\leq{t}\leq{n-k-1} $, define $f_{t,s}\in{\mathbb{F}_q}$ for $0\leq{s}\leq{k-1}$ by the following
\begin{equation}\label{S10}
(a_{t_1}^{k+t},a_{t_2}^{k+t},\ldots,a_{t_k}^{k+t})=
(f_{t,0},f_{t,1},\ldots,f_{t,k-1})
\begin{pmatrix}
1&1&\cdots&1\\
a_{t_1}&a_{t_2}&\cdots&a_{t_k}\\
\vdots&\vdots&\ddots&\vdots\\
a_{t_1}^{k-1}&a_{t_2}^{k-1}&\cdots&a_{t_k}^{k-1}\\
\end{pmatrix},
\end{equation}
where $f_{t,s}$'s are determined by $a_{t_i}$'s and $t$.
%then $a_i^{k+t}=\sum\limits_{s=0}^{k-1}f_{t,s}a_i^s$, where $f_{t,s}\in\mathbb{F}_q$ for all $i\in\{1,...,k\}$, then
% where $f_{t,s}\in{\mathbb{F}_q}$ for $0\leq{s}\leq{k-1}$.
Then
\begin{equation}\label{S11}
% {\color{red}f_{t,s}=-\sum_{0\leq i \leq t, 0\leq j\leq k,j-i=k-s}c_je_{t-i},\, 0\leq s\leq k-1.}
f_{t,s}=-\sum_{i=0}^{\min\{t,s\}}c_{i+k-s}e_{t-i},\, 0\leq s\leq k-1,
\end{equation}
where $e_0=1$ and $e_i=-\sum_{j=0}^{i-1}e_jc_{i-j}$ for $0\leq{i}\leq{t}$.
\end{lem}

\begin{proof}
From \eqref{S10}, we have $a_{t_i}^{k+t}=\sum_{s=0}^{k-1}f_{t,s}a_{t_i}^s$ for $1\leq i \leq k$.
Therefore, $a_{t_1},a_{t_2},\ldots,a_{t_k}$ are zeros of the polynomial $f^{(t)}(x)=x^{k+t}-\sum_{s=0}^{k-1}f_{t,s}x^s$. Note that
$a_{t_1},a_{t_2},\ldots,a_{t_k}$ are also zeros of the polynomial $g(x)=\sum_{j=0}^{k}c_jx^{k-j}=\prod_{i=1}^{k}(x-a_{t_i})$,
and $\deg(g(x))\leq \deg(f^{(t)}(x))$. Then it follows that $g(x)$ divides $f^{(t)}(x)$.
Hence, there exists some $h^{(t)}(x)=\sum_{i=0}^{t}w_i^{(t)}x^i \in \fq[x]$, where $w_i^{(t)}\in \fq$, such that
\begin{equation}\label{S12}
f^{(t)}(x)=g(x)h^{(t)}(x)=(\sum_{j=0}^{k}c_jx^{k-j})(\sum_{i=0}^{t}w_i^{(t)}x^i).
\end{equation}

Observe that in the polynomial $f^{(t)}(x)$, the coefficient of the term with degree $k+t$ is $1$, and all the coefficients of the terms with degree less than $k+t$ but greater than $k-1$ are $0$. It then follows from \eqref{S12} that
\begin{equation*}
(0,0,...,1)=
(w_0^{(t)},w_1^{(t)},...,w_{t}^{(t)})
\begin{pmatrix}
  c_0&0&\cdots&0\\
  c_1&c_0&\cdots&0\\
  \vdots&\vdots&\ddots&\vdots\\
  c_t&c_{t-1}&\cdots&c_0
  \end{pmatrix}.
\end{equation*}
Therefore we have $(w_0^{(t)},w_1^{(t)},...,w_{t}^{(t)})=(0,0,...,1)A_t^{-1}$, where
\[
A_t=\begin{pmatrix}
c_0&0&\cdots&0\\
c_1&c_0&\cdots&0\\
\vdots&\vdots&\ddots&\vdots\\
c_t&c_{t-1}&\cdots&c_0
\end{pmatrix}.
\]
By Lemma \ref{L3}, we have
\begin{equation*}
(w_0^{(t)},w_1^{(t)},...,w_{t}^{(t)})=(0,0,...,1)A_t^{-1}=(e_t,e_{t-1},...,e_0).
\end{equation*}
Thus $w_i^{(t)}=e_{t-i}$ for $0\leq{i}\leq{t}$. By comparing the coefficients of terms with degree $\leq k-1$ on both sides of \eqref{S12}, we obtain
\begin{equation*}
f_{t,s}=-\sum_{i=0}^{\min\{t,s\}}c_{i+k-s}w_i^{(t)}=-\sum_{i=0}^{\min\{t,s\}}c_{i+k-s}e_{t-i}, \,\,0\leq{s}\leq{k-1}.
\end{equation*}
This completes the proof.
\end{proof}

For given $\boldsymbol{\alpha}=(a_1,a_2,...,a_n)\in\mathbb{F}_{q}^n$ with distinct $a_i$'s and $k$-subset $\mathcal{T}=\{t_1,\ldots,t_k\}$ of $\{1,...,n\}$, the set $\{a_{t_i}:i\in \mathcal{T}\}$ defines a matrix $F_\mathcal{T}$ over $\fq$ given by 
\begin{equation}\label{S13}
F_\mathcal{T}=\begin{pmatrix}
f_{0,0}&f_{0,1}&\cdots&f_{0,k-1}\\
f_{1,0}&f_{1,1}&\cdots&f_{1,k-1}\\
\vdots&\vdots&\ddots&\vdots\\
f_{n-k-1,0}&f_{n-k-1,1}&\cdots&f_{n-k-1,k-1}\\
\end{pmatrix},
\end{equation}
where $f_{t,s}$'s  are defined by  \eqref{S11} for $0\leq{t}\leq{n-k-1} $ and $0\leq{s}\leq{k-1}$. 

In the following, we investigate the MDS property of $(\mL,\mP)$-TGRS codes for the most general case where $\mathcal{L}=[n-k-1]$ and $\mathcal{P}=[k-1]$.
\begin{thm}\label{T1}
Let $\boldsymbol{\alpha}=(a_1,\ldots,a_n)\in\mathbb{F}_{q}^n$ with distinct $a_i$'s, $\boldsymbol{\nu}=(v_1,\ldots,v_n)\in(\mathbb{F}_{q}^*)^n$ and $B=[b_{i,j}] \in \fq^{k\times{(n-k)}}$ be as in \eqref{S8-B}. Let $I_k$ be the $k\times{k}$ identity matrix over $\fq$ and $F_\mathcal{T}$ be the $(n-k)\times k$ matrix defined as in \eqref{S13}. Then the $(\mL,\mP)$-TGRS code $\C(\mL,\mP,B)$ defined by \eqref{S6} is an MDS code if and only if $B\in\Omega$, where
\begin{equation} \label{Omega}
\Omega:=\{B\in\mathbb{F}_{q}^{k\times(n-k)}:\, \vert{I_k+BF_\mathcal{T}}\vert\neq0 \mbox{ for all $k$-subset } \mathcal{T}\subseteq\{1,\ldots,n\} \}.
\end{equation}
\end{thm}

\begin{proof}
Up to the equivalence of codes, we always assume that $\boldsymbol{\nu}=(1,\ldots,1)$ in the proof.
By Lemma \ref{L2}, $\C(\mL,\mP,B)$ is an MDS code if and only if all $k\times{k}$ minors of the generator matrix $G_{TGRS}$ in \eqref{S8-GM} are nonzero. Then $\C(\mL,\mP,B)$ is MDS if and only if the determinant of the matrix generated by any $k$ columns of $G_{TGRS}$ is nonzero. Let $\mathcal{T}:=\{t_1,...,t_k\}$ be a $k$-subset of $\{1,...,n\}$. Then $\mathcal{T}$ corresponds to the index set of the $k$ columns of $G_{TGRS}$. It should be noted that the only difference between the columns of $G_{TGRS}$ lies in $a_i$, where $1\leq i \leq n$. Without loss of generality, we focus on the first $k$ columns of $G_{TGRS}$, namely, the case $\mathcal{T}=\{1,...,k\}$. Then the $k\times{k}$ submatrix of $G_{TGRS}$ corresponding to $\mathcal{T}$ is given by
\begin{equation}\label{S15}
G_\mathcal{T}=\begin{pmatrix}
  1+\sum\limits_{j=0}^{n-k-1}b_{0,j}a_1^{k+j}&\cdots&1+\sum\limits_{j=0}^{n-k-1}b_{0,j}a_k^{k+j}\\
a_1+\sum\limits_{j=0}^{n-k-1}b_{1,j}a_1^{k+j}&\cdots&a_k+\sum\limits_{j=0}^{n-k-1}b_{1,j}a_k^{k+j}\\
\vdots&\vdots&\vdots\\
a_1^{k-1}+\sum\limits_{j=0}^{n-k-1}b_{k-1,j}a_1^{k+j}&\cdots&a_k^{k-1}+\sum\limits_{j=0}^{n-k-1}b_{k-1,j}a_k^{k+j}\\
  \end{pmatrix}.
\end{equation}

Next, we will compute the determinant of $G_\mathcal{T}$.
% \begin{equation*}
% %\begin{flushleft}
% \vert{G_\mathcal{T}}\vert= \begin{vmatrix}
% 1+\sum\limits_{j=0}^{n-k-1}b_{0,j}a_1^{k+j}&\cdots&1+\sum\limits_{j=0}^{n-k-1}b_{0,j}a_k^{k+j}\\
% a_1+\sum\limits_{j=0}^{n-k-1}b_{1,j}a_1^{k+j}&\cdots&a_k+\sum\limits_{j=0}^{n-k-1}b_{1,j}a_k^{k+j}\\
% \vdots&\vdots&\vdots\\
% a_1^{k-1}+\sum\limits_{j=0}^{n-k-1}b_{k-1,j}a_1^{k+j}&\cdots&a_k^{k-1}+\sum\limits_{j=0}^{n-k-1}b_{k-1,j}a_k^{k+j}\\
% \end{vmatrix}\\
% %\end{flushleft}
% %\begin{flushleft}
% \end{equation*}
By Lemma \ref{L4}, for $1\leq i \leq k$ and $0\leq{j}\leq{n-k-1}$, the terms $a_i^{k+j}$  can be expressed as $a_i^{k+j}=\sum_{s=0}^{k-1}f_{j,s}a_i^s$, where $f_{j,s}$ is given as in \eqref{S11}. It then follows that
\begin{equation*}
\vert{G_\mathcal{T}}\vert=\begin{vmatrix}
  1+\sum\limits_{s=0}^{k-1}\sum\limits_{j=0}^{n-k-1}b_{0,j}f_{j,s}a_1^s&\cdots&1+\sum\limits_{s=0}^{k-1}\sum\limits_{j=0}^{n-k-1}b_{0,j}f_{j,s}a_k^s\\
a_1+\sum\limits_{s=0}^{k-1}\sum\limits_{j=0}^{n-k-1}b_{1,j}f_{j,s}a_1^s&\cdots&a_k+\sum\limits_{s=0}^{k-1}\sum\limits_{j=0}^{n-k-1}b_{1,j}f_{j,s}a_k^s\\
\vdots&\vdots&\vdots\\
a_1^{k-1}+\sum\limits_{s=0}^{k-1}\sum\limits_{j=0}^{n-k-1}b_{k-1,j}f_{j,s}a_1^s&\cdots&a_k^{k-1}+\sum\limits_{s=0}^{k-1}\sum\limits_{j=0}^{n-k-1}b_{k-1,j}f_{j,s}a_k^s
\end{vmatrix}.
\end{equation*}
%\end{flushleft}
%\begin{flushleft}
By decomposing the matrix corresponding to the determinant, we have
\begin{equation*}
\vert{G_\mathcal{T}}\vert=\begin{vmatrix}
1+\sum\limits_{j=0}^{n-k-1}b_{0,j}f_{j,0}&\sum\limits_{j=0}^{n-k-1}b_{0,j}f_{j,1}&\cdots&\sum\limits_{j=0}^{n-k-1}b_{0,j}f_{j,k-1}\\
\sum\limits_{j=0}^{n-k-1}b_{1,j}f_{j,0}&1+\sum\limits_{j=0}^{n-k-1}b_{1,j}f_{j,1}&\cdots&\sum\limits_{j=0}^{n-k-1}b_{1,j}f_{j,k-1}\\
\vdots&\vdots&\ddots&\vdots\\
\sum\limits_{j=0}^{n-k-1}b_{k-1,j}f_{j,0}&\sum\limits_{j=0}^{n-k-1}b_{k-1,j}f_{j,1}&\cdots&1+\sum\limits_{j=0}^{n-k-1}b_{k-1,j}f_{j,k-1}
\end{vmatrix}
%\end{flushleft}
%\begin{flushleft}
\cdot\begin{vmatrix}
1&1&\cdots&1\\
a_1&a_2&\cdots&a_k\\
a_1^2&a_2^2&\cdots&a_k^2\\
  \vdots&\vdots&\ddots&\vdots\\
a_{1}^{k-1}&a_{2}^{k-1}&\cdots&a_{k}^{k-1}\\
\end{vmatrix}.
\end{equation*}
Note that the matrix on the right-hand side of $G_\mathcal{T}$ with respect to $a_i$'s is a $k\times{k}$ Vandermonde determinant. Thus it leads to
\begin{equation*}
\vert{G_\mathcal{T}}\vert=\begin{vmatrix}
1+\sum\limits_{j=0}^{n-k-1}b_{0,j}f_{j,0}&\sum\limits_{j=0}^{n-k-1}b_{0,j}f_{j,1}&\cdots&\sum\limits_{j=0}^{n-k-1}b_{0,j}f_{j,k-1}\\
\sum\limits_{j=0}^{n-k-1}b_{1,j}f_{j,0}&1+\sum\limits_{j=0}^{n-k-1}b_{1,j}f_{j,1}&\cdots&\sum\limits_{j=0}^{n-k-1}b_{1,j}f_{j,k-1}\\
\vdots&\vdots&\ddots&\vdots\\
\sum\limits_{j=0}^{n-k-1}b_{k-1,j}f_{j,0}&\sum\limits_{j=0}^{n-k-1}b_{k-1,j}f_{j,1}&\cdots&1+\sum\limits_{j=0}^{n-k-1}b_{k-1,j}f_{j,k-1}\\
\end{vmatrix}\cdot{\prod_{1\leq{j}<{i}\leq{k}}(a_i-a_j)}.
%\end{flushleft}
\end{equation*}
One can check that the remaining determinant as above can be expressed as $\vert{I_k+BF_\mathcal{T}}\vert$, where $F_\mathcal{T}$ is the matrix defined as in \eqref{S13}. Then it gives
\begin{equation*}
\vert{G_\mathcal{T}}\vert=\vert{I_k+BF_\mathcal{T}}\vert\cdot{\prod_{1\leq{j}<{i}\leq{k}}(a_i-a_j)}.
\end{equation*}
It is clear that $\vert{G_\mathcal{T}}\vert\neq{0}$ if and only if $\vert{I_k+BF_\mathcal{T}}\vert\neq{0}$ since ${\prod_{1\leq{j}<{i}\leq{k}}(a_i-a_j)}$ is nonzero.

With the discussion as above, we conclude that the code $\C(\mL,\mP,B)$ is MDS if and only if $B\in\Omega$, where $\Omega$ is given by \eqref{Omega}. This completes the proof.
\end{proof}

\begin{remark}\label{R4}
In Theorem \ref{T1}, we provide a necessary and sufficient condition for $(\mL,\mP)$-TGRS codes of the most general form to be MDS via the coefficient matrix $B$. It should be noted that by selecting specific coefficient matrix $B$, we can reproduce the main results on MDS property of TGRS codes in the previous works \cite{ZHAO,BEPU,SUIL,GU,DING}. Moreover, the condition presented in our Theorem \ref{T1} is more concise and simpler than that of Zhao et al. \cite{ZHAO}, and the proof of Theorem \ref{T1} is different from that of Zhao et al. and is much shorter and more efficient.
% {\color{blue}Zhao et al. provide a comprehensive analysis of the conditions under which TGRS codes with arbitrary twists are MDS \cite{ZHAO}. In contrast, this article uses a more direct proof method to explore this property, leading to a more concise theorem. From this theorem, several effective corollaries and examples have also been derived.}
\end{remark}

\begin{remark}
When $B=\textbf{0}$, the $(\mL,\mP)$-TGRS code $\C(\mL,\mP,B)$ is reduced to a GRS code, and $\vert{I_k+BF_\mathcal{T}}\vert=\vert{I_k}\vert=1$ for all $k$-subset $\mathcal{T}$ which implies that it is MDS directly.
\end{remark}

% Recall the Definition \ref{D2} of TGRS codes. We obtain a TGRS code with arbitrary distortion, as shown in (\ref{S6}), by taking the maximum values of the twists set and the position set. Now, by selecting different coefficient matrices $B$, we can derive the TGRS codes discussed earlier. Furthermore, by substituting the appropriate coefficient matrix $B$ into Theorem \ref{T1}, we can also recover the conclusions about MDS TGRS codes presented in the previous literature.

In the following, we give some corollaries for Theorem \ref{T1} by selecting specific coefficient matrix $B$.
\begin{cor}\label{Cor1}
Let
\begin{equation*}
B=\begin{pmatrix}
  b_{0,0}&0&\cdots&0\\
  0&0&\dots&0\\
  \vdots&\vdots&\ddots&\vdots\\
  0&0&\dots&0\\
\end{pmatrix}
\end{equation*}
and $\nu=(1,...,1)$. Then the $(\mL,\mP)$-TGRS code $\C(\mL,\mP,B)$ in Theorem \ref{T1} is MDS if and only if for any $k$-subset $\mathcal{T}\subseteq\{1,...,n\}$ we have $b_{0,0}(-1)^k\prod_{i\in{\mathcal{T}}}a_i\neq{1}$, which was given in \cite[Lemma 4]{BEPU}.
% we obtain the TRS codes \cite{BEPU} for $(t,h)=(0,1)$. The code $\C_{n,k}(\boldsymbol{\alpha},\textbf{1},B)$ is MDS if and only if
% \begin{equation*}
% \mbox{for}\ \mbox{all}\ k\mbox{-}\mbox{subset}\ \mathcal{T}\subseteq\{1,...,n\},\ b_{0,0}(-1)^k\prod_{i\in{\mathcal{T}}}a_i\neq{1}.
% \end{equation*}
% \begin{proof}
% For every $k$$\mbox{-}$subset $\mathcal{T}\subseteq\{1,...,n\}$, we have
% \begin{equation*}
% \vert{I_k+BF_\mathcal{T}}\vert=\vert{1+b_{0,0}f_{0,0}}\vert=1-b_{0,0}(-1)^k\prod\limits_{i\in{\mathcal{T}}}a_i\neq{0},
% \end{equation*}
% \end{proof}
%   Therefore, we obtain the \cite[Lemma 4]{BEPU}.
\end{cor}

\begin{cor}\label{Cor2}
Let
\begin{equation*}
B=\begin{pmatrix}
  0&0&\cdots&0\\
  0&0&\dots&0\\
  \vdots&\vdots&\ddots&\vdots\\
  b_{k-1,0}&0&\dots&0\\
\end{pmatrix}
\end{equation*}
and $\nu=(1,...,1)$.
Then the $(\mL,\mP)$-TGRS code $\C(\mL,\mP,B)$ in Theorem \ref{T1} is MDS if and only if for any $k$-subset $\mathcal{T}\subseteq\{1,...,n\}$ we have $b_{k-1,0}\sum_{i\in{\mathcal{T}}}a_i\neq{-1}$, which was given in \cite[Lemma 10]{BEPU}.
% Then, we obtain the TRS Codes \cite{BEPU} for $(t,h)=(1,k-1)$. The code $\C_{n,k}(\boldsymbol{\alpha},\textbf{1},\boldsymbol{B})$ is MDS if and only if
% \begin{equation*}
% \mbox{for}\ \mbox{all}\ k\mbox{-}\mbox{subset}\ \mathcal{T}\subseteq\{1,...,n\},\ b_{k-1,0}\sum_{i\in{\mathcal{T}}}a_i\neq{-1}
% \end{equation*}
% \begin{proof}
% For every $k$$\mbox{-}$subset $\mathcal{T}\subseteq\{1,...,n\}$, we have
% \begin{equation*}
% \vert{I_k+BF_\mathcal{T}}\vert=\vert{1+b_{k-1,0}f_{0,k-1}}\vert=1-b_{k-1,0}\sum\limits_{i\in{\mathcal{T}}}(-a_i)\neq0,
% \end{equation*}
% \end{proof}
% Therefore, we obtain the \cite[Lemma 10]{BEPU}.
\end{cor}

\begin{cor}\label{Cor3}
Let
\begin{equation*}
B=\begin{pmatrix}
  0&0&\cdots&0\\
  \vdots&\vdots&&\vdots\\
  0&b_{k-2,1}&\dots&0\\
  b_{k-1,0}&0&\dots&0
  \end{pmatrix}.
\end{equation*}
Then the $(\mL,\mP)$-TGRS code $\C(\mL,\mP,B)$ in Theorem \ref{T1} is MDS if and only if for any $k$-subset $\mathcal{T}\subseteq\{1,...,n\}$ we have $\prod_{i\in{\mathcal{T}}}(x-a_i)=\sum_{j=0}^kc_jx^{k-j}$, which was given in \cite[Theorem 3.3]{SUIL}.
% Then, we obtain the TGRS codes \cite{SUIL}. The code $\C_{n,k}(\boldsymbol{\alpha},\boldsymbol{\nu},\boldsymbol{B})$ is MDS if and only if for all $k$$\mbox{-}$subset $\mathcal{T}\subseteq\{1,...,n\}$,
% \begin{equation*}
% \prod_{i\in{\mathcal{T}}}(x-a_i)=\sum_{j=0}^kc_jx^{k-j},
% \end{equation*}
% and
% \begin{equation*}
% 1-c_1b_{k-1,0}-b_{k-2,1}(c_3-c_1c_2)+b_{k-1,0}b_{k-2,1}(c_1c_3-c_2^2)\neq0.
% \end{equation*}
% \begin{proof}
% For every $k$$\mbox{-}$subset $\mathcal{T}\subseteq\{1,...,n\}$, we have
% $$
% \begin{aligned}
% \vert{I_k+BF_\mathcal{T}}\vert&=(1+b_{k-1,0}f_{0,k-1})(1+b_{k-2,1}f_{1,k-2})-b_{k-2,1}f_{1,k-1}b_{k-1,0}f_{0,k-2}\\
% &=1-c_1b_{k-1,0}-b_{k-2,1}(c_3-c_1c_2)+b_{k-1,0}b_{k-2,1}(c_1c_3-c_2^2)\neq0
%   \end{aligned}
% $$
% \end{proof}
% Therefore, we obtain the \cite[Theorem 3.3]{SUIL}.
\end{cor}

\begin{cor}\label{Cor4}
Let
\begin{equation*}
B=\begin{pmatrix}
\boldsymbol{0}_{(k-\ell)\times{\ell}}&\boldsymbol{0}_{(k-\ell)\times{(n-k-\ell)}}\\
\boldsymbol{D}_{\ell\times{\ell}}&\boldsymbol{0}_{\ell\times{(n-k-\ell)}}\\
\end{pmatrix},
D=\begin{pmatrix}
  b_{k-\ell,0}&&&\\
  &b_{k-\ell+1,1}&&\\
  &&\ddots&\\
  &&&b_{k-\ell,\ell-1}\\
  \end{pmatrix}
\end{equation*}
% where
% \begin{equation*}
% D=\begin{pmatrix}
% b_{k-\ell,0}&&&\\
% &b_{k-\ell+1,1}&&\\
% &&\ddots&\\
% &&&b_{k-\ell,\ell-1}\\
% \end{pmatrix},
% \end{equation*}
and $\ell<min\{k,n-k\}$.
Then the $(\mL,\mP)$-TGRS code $\C(\mL,\mP,B)$ in Theorem \ref{T1} is MDS if and only if $D\in\Omega$, where
\begin{equation*}
\Omega=\{D\in{\mathbb{F}_q^{\ell\times{\ell}}}: M(D,\boldsymbol{\alpha},\mathcal{T},\ell)\neq{0} \mbox{ for all $k$-subset } \mathcal{T}\subseteq\{1,...,n\}\}
\end{equation*}
and
$$
\begin{aligned}
M(D,\boldsymbol{\alpha},\mathcal{T},\ell)
&=\begin{vmatrix}
1+b_{k-\ell,0}f_{0,k-\ell}&b_{k-\ell,0}f_{0,k-\ell+1}&\cdots&b_{k-\ell,0}f_{0,k-\ell}\\
b_{k-\ell+1,1}f_{1,k-\ell}&1+b_{k-\ell+1,1}f_{1,k-\ell+1}&\cdots&b_{k-\ell+1,1}f_{1,k-\ell}\\
\vdots&\vdots&\ddots&\vdots\\
b_{k-1,\ell-1}f_{\ell-1,k-\ell}&b_{k-1,\ell-1}f_{\ell-1,k-\ell+1}&\cdots&1+b_{k-1,\ell-1}f_{\ell-1,k-\ell}
\end{vmatrix}=\vert{I_k+BF_\mathcal{T}}\vert.
\end{aligned}
$$
This was also given in \cite[Theorem III.3]{GU}.
\end{cor}

\begin{cor}\label{Cor5}
Let
\[B=\begin{pmatrix}
\boldsymbol{0}_{(k-\ell)\times{\ell}}&\boldsymbol{0}_{(k-\ell)\times{(n-k-\ell)}}\\
\boldsymbol{A}_{\ell\times{\ell}}&\boldsymbol{0}_{\ell\times{(n-k-\ell)}}\\
\end{pmatrix},
A=\begin{pmatrix}
  b_{k-\ell,0}&b_{k-\ell,1}&\cdots&b_{k-\ell,\ell-1}\\
  b_{k-\ell+1,0}&b_{k-\ell+1,1}&\cdots&b_{k-\ell+1,\ell-1}\\
  \vdots&\vdots&\ddots&\vdots\\
  b_{k-1,0}&b_{k-1,1}&\cdots&b_{k-1,\ell-1}\\
  \end{pmatrix}
\]
% where
% $
% A=\begin{pmatrix}
% b_{k-\ell,0}&b_{k-\ell,1}&\cdots&b_{k-\ell,\ell-1}\\
% b_{k-\ell+1,0}&b_{k-\ell+1,1}&\cdots&b_{k-\ell+1,\ell-1}\\
% \vdots&\vdots&\ddots&\vdots\\
% b_{k-1,0}&b_{k-1,1}&\cdots&b_{k-1,\ell-1}\\
% \end{pmatrix}
% $
and $\ell<min\{k,n-k\}$.
Then the $(\mL,\mP)$-TGRS code $\C(\mL,\mP,B)$ in Theorem \ref{T1} is MDS if and only if $A\in\Omega$, where
\begin{equation*}
\Omega=\{A\in{\mathbb{F}_q^{\ell\times{\ell}}}: \Psi(A,\boldsymbol{\alpha},\mathcal{T},\ell)\neq{0} \mbox{ for all $k$-subset }  \mathcal{T}\subseteq\{1,...,n\}\}
\end{equation*}
%$\forall{\mbox{$k$-subset $I\subseteq\{1,...,n\}$}}$
and
$$
\begin{aligned}
\Psi(A,\boldsymbol{\alpha},\mathcal{T},\ell)&=\begin{vmatrix}
1+\sum\limits_{i=0}^{\ell-1}b_{k-\ell,i}f_{i,k-\ell}&\sum\limits_{i=0}^{\ell-1}b_{k-\ell,i}f_{i,k-\ell+1}&\cdots&\sum\limits_{i=0}^{\ell-1}b_{k-\ell,i}f_{i,k-1}\\
\sum\limits_{i=0}^{\ell-1}b_{k-\ell+1,i}f_{i,k-\ell}&1+\sum\limits_{i=0}^{\ell-1}b_{k-\ell+1,i}f_{i,k-\ell+1}&\cdots&\sum\limits_{i=0}^{\ell-1}b_{k-\ell+1,i}f_{i,k-1}\\
\vdots&\vdots&\ddots&\vdots\\
\sum\limits_{i=0}^{\ell-1}b_{k-1,i}f_{i,k-\ell}&\sum\limits_{i=0}^{\ell-1}b_{k-1,i}f_{i,k-\ell+1}&\cdots&1+\sum\limits_{i=0}^{\ell-1}b_{k-1,i}f_{i,k-1}\\
\end{vmatrix}=\vert{I_k+BF_\mathcal{T}}\vert.&
\end{aligned}
$$
This was also given in \cite[Theorem 3.2]{DING}.
\end{cor}

\begin{example}\label{EX1}
Let $n=6$, $k=4$, $q=7$, $\mathbb{F}_7=\{0,1,2,3,4,5,6\}$, $\boldsymbol{\alpha}=(1,2,3,4,5,6)\in{\mathbb{F}_7^6}$, $\boldsymbol{\nu}=(1,\ldots,1)$ and $B=[b_{i,j}] \in \fq^{4 \times 2}$.
% Let matrix $B$ be
% \begin{equation*}
% B=\begin{pmatrix}
% b_{00}&b_{01}\\
% b_{10}&b_{11}\\
% b_{20}&b_{21}\\
% b_{30}&b_{31}\\
% \end{pmatrix}_{4\times{2}}
% \end{equation*}
% and
% Let
% \begin{equation*}
% \mathcal{P}(6,4,\boldsymbol{B})=\{\sum\limits_{i=0}^3f_ix^i+\sum\limits_{i=0}^3f_i\sum\limits_{j=0}^1b_{ij}x^{4+j}:\ \forall{f_i}\in{\mathbb{F}_7},\ 0\leq{i}\leq3\}.
% \end{equation*}
Recall from Theorem \ref{T1} that the $(\mL,\mP)$-TGRS code is MDS if and only if $B\in{\Omega}$. Magma experiments shows that $\Omega$ is given by
\[\Omega=\left\{\begin{pmatrix}
4&6\\
5&5\\
5&2\\
4&0\\
\end{pmatrix},
\begin{pmatrix}
3&5\\
4&3\\
2&1\\
6&3\\
\end{pmatrix},
\begin{pmatrix}
1&4\\
1&1\\
3&5\\
0&1\\
\end{pmatrix},
\begin{pmatrix}
3&3\\
0&4\\
1&1\\
4&5\\
\end{pmatrix},
\begin{pmatrix}
6&5\\
1&1\\
6&6\\
0&6\\
\end{pmatrix},
\begin{pmatrix}
0&5\\
5&5\\
3&2\\
1&1\\
\end{pmatrix},
\begin{pmatrix}
1&1\\
6&5\\
4&2\\
6&0\\
\end{pmatrix},
\begin{pmatrix}
0&6\\
6&3\\
2&6\\
2&1\\
\end{pmatrix},
\begin{pmatrix}
3&2\\
3&1\\
4&0\\
4&6\\
\end{pmatrix},\ldots
\right\},
\]
where the cardinality of $\Omega$ is $390841$. When $B\in \Omega$, $\C(\mL,\mP,B)$ is a $[6,4,3]$ MDS code.

% Specially, let
% \begin{equation*}
% B=\begin{pmatrix}
% 1&1\\
% 6&5\\
% 4&2\\
% 6&0\\
% \end{pmatrix}.
% \end{equation*} In this case,
% \begin{equation*}
% \mathcal{P}(6,4,\boldsymbol{B})=\{\sum\limits_{i=0}^3f_ix^i+(f_0+6f_1+4f_2+6f_3)x^4+(f_0+5f_1+2f_2)x^5:\ \forall{f_i}\in{\mathbb{F}_7},\ 0\leq{i}\leq3\},
% \end{equation*}
% and $\C(\mL,\mP,B)$ is a $[6,4,3]$ MDS code.
\end{example}

\begin{example}\label{EX2}
Let $n=6$, $k=3$, $q=7$, $\mathbb{F}_7=\{0,1,2,3,4,5,6\}$, $\boldsymbol{\alpha}=(1,2,3,4,5,6)\in{\mathbb{F}_7^6}$, $\boldsymbol{\nu}=(1,\ldots,1)$ and $B=[b_{i,j}] \in \fq^{3 \times 3}$.
% Let
% % matrix $B$ be
% % \begin{equation*}
% % B=\begin{pmatrix}
% % b_{00}&b_{01}&b_{02}\\
% % b_{10}&b_{11}&b_{12}\\
% % b_{20}&b_{21}&b_{22}\\
% % \end{pmatrix}_{3\times{3}}
% % \end{equation*}
% % and
% \begin{equation*}
% \mathcal{P}(6,3,\boldsymbol{B})=\{\sum\limits_{i=0}^2f_ix^i+\sum\limits_{i=0}^2f_i\sum\limits_{j=0}^2b_{ij}x^{3+j}:\ \forall{f_i}\in{\mathbb{F}_7},\ 0\leq{i}\leq3\}.
% \end{equation*}
Recall from Theorem \ref{T1} that the $(\mL,\mP)$-TGRS code is MDS if and only if $B\in{\Omega}$. Magma experiments shows that $\Omega$ is given by
\begin{equation*}
\Omega=\left\{\begin{pmatrix}
2&5&3\\
2&1&1\\
3&2&2\\
\end{pmatrix},
\begin{pmatrix}
3&4&4\\
0&0&3\\
4&0&2\\
\end{pmatrix},
\begin{pmatrix}
5&1&6\\
5&2&4\\
3&1&0\\
\end{pmatrix},
\begin{pmatrix}
3&2&4\\
2&1&3\\
6&0&5\\
\end{pmatrix},
\begin{pmatrix}
1&6&6\\
5&0&2\\
5&2&4\\
\end{pmatrix},\ldots
\right\}.
\end{equation*}
where the cardinality of $\Omega$ is $894747$. When $B\in \Omega$, $\C(\mL,\mP,B)$ is a $[6,3,4]$ MDS code.

% Specially, let
% \begin{equation*}
% B=\begin{pmatrix}
% 2&5&3\\
% 2&1&1\\
% 3&2&2\\
% \end{pmatrix}.
% \end{equation*} In this case,
% \begin{equation*}
% \mathcal{P}(6,3,\boldsymbol{B})=\{\sum\limits_{i=0}^2f_ix^i+(2f_0+2f_1+3f_2)x^3+(5f_0+f_1+2f_2)x^4+(3f_0+f_1+2f_2)x^5:\ \forall{f_i}\in{\mathbb{F}_7},\ 0\leq{i}\leq3\},
% \end{equation*}
% and $\C(\mL,\mP,B)$ is a $[6,3,4]$ MDS code.
\end{example}

%\begin{example}\label{EX3}
%Let $n=6$, $k=3$, $q=9$, $\mathbb{F}_9^*=\langle z\rangle $, $\boldsymbol{\alpha}=\{1,2,z,z^2,z^3,z^5,z^6,z^7\}$ and $B=[b_{i,j}] \in \fq^{3 \times 3}$.
% Let matrix $B$ be
% \begin{equation*}
% B=\begin{pmatrix}
% b_{00}&b_{01}&b_{02}\\
% b_{10}&b_{11}&b_{12}\\
% b_{20}&b_{21}&b_{22}\\
% \end{pmatrix}_{3\times{3}}
% \end{equation*}
% and
% Let
% \begin{equation*}
% \mathcal{P}(8,3,\boldsymbol{B})=\{\sum\limits_{i=0}^2f_ix^i+\sum\limits_{i=0}^2f_i\sum\limits_{j=0}^2b_{ij}x^{3+j}:\ \forall{f_i}\in{\mathbb{F}_9},\ 0\leq{i}\leq3\}.
% \end{equation*}
%Recall from Theorem \ref{T1} that the $(\mL,\mP)$-TGRS code is MDS if and only if $B\in{\Omega}$. Magma experiments show that $\Omega$ is given
%\begin{equation*}
%\Omega=\left\{\begin{pmatrix}
%z^3&z^3&z^6\\
%z^3&1&1\\
%z^7&z&2\\
%\end{pmatrix},
%\begin{pmatrix}
%z^3&1&z^2\\
%z&z^6&z^3\\
%z^2&z^2&z^7\\
%\end{pmatrix},
%\begin{pmatrix}
%z^3&1&z^2\\
%1&z^6&z^3\\
%1&z^6&0\\
%\end{pmatrix},
%\begin{pmatrix}
%2&0&z^6\\
%2&1&z^6\\
%z^3&z&z\\
%\end{pmatrix},
%\begin{pmatrix}
%z&z^5&z^3\\
%2&z^3&z^7\\
%z^2&2&2\\
%\end{pmatrix},...
%\right\}.
%\end{equation*}
%where the cardinality of $\Omega$ is $24977$. When $B\in \Omega$, $\C(\mL,\mP,B)$ is a $[8,3,6]$ MDS code.
%\end{example}

\begin{example}\label{EX3}
Let $n=8$, $k=3$, $q=9$, $\mathbb{F}_9^*=\langle z\rangle $, $\boldsymbol{\alpha}=\{1,2,z,z^2,z^3,z^5,z^6,z^7\}$, $\boldsymbol{\nu}=(1,\ldots,1)$ and
$B$ be of the form 
\begin{equation*}
B=\begin{pmatrix}
b_{00}&b_{01}&b_{02}&0&0\\
b_{10}&b_{11}&b_{12}&0&0\\
b_{20}&b_{21}&b_{22}&0&0
\end{pmatrix}.
\end{equation*}
Recall from Theorem \ref{T1} that the $(\mL,\mP)$-TGRS code is MDS if and only if $B\in{\Omega}$. Magma experiments show that $\Omega$ is given
\begin{equation*}
\Omega=\left\{\begin{pmatrix}
z^3&z^3&z^6&0&0\\
z^3&1&1&0&0\\
z^7&z&2&0&0
\end{pmatrix},
\begin{pmatrix}
z^3&1&z^2&0&0\\
z&z^6&z^3&0&0\\
z^2&z^2&z^7&0&0
\end{pmatrix},
\begin{pmatrix}
z^3&1&z^2&0&0\\
1&z^6&z^3&0&0\\
1&z^6&0&0&0
\end{pmatrix},
\begin{pmatrix}
2&0&z^6&0&0\\
2&1&z^6&0&0\\
z^3&z&z&0&0
\end{pmatrix},
%\begin{pmatrix}
%z&z^5&z^3&0&0\\
%2&z^3&z^7&0&0\\
%z^2&2&2&0&0\\
%\end{pmatrix},
\ldots
\right\}.
\end{equation*}
where the cardinality of $\Omega$ is $24977$. When $B\in \Omega$, $\C(\mL,\mP,B)$ is an $[8,3,6]$ MDS code.
\end{example}

% Specially, let
% \begin{equation*}
% B=\begin{pmatrix}
% z^3&1&z^2\\
% z&z^6&z^3\\
% z^2&z^2&z^7\\
% \end{pmatrix}.
% \end{equation*} In this case,
% \begin{equation*}
% \mathcal{P}(8,3,\boldsymbol{B})=\{\sum\limits_{i=0}^2f_ix^i+(z^3f_0+zf_1+z^2f_2)x^3+(f_0+z^6f_1+z^2f_2)x^4+(z^2f_0+z^3f_1+z^7f_2)x^5:\ \forall{f_i}\in{\mathbb{F}_9},\ 0\leq{i}\leq3\}
% \end{equation*}
% and $\C(\mL,\mP,B)$ is a $[8,3,6]$ MDS code.

% After investigating the necessary and sufficient conditions under which TGRS codes are MDS codes,

We now turn our attention to the condition under which the $(\mL,\mP)$-TGRS code is NMDS.
% An $[n,k]$ linear code $\C$ over $\fq$ is called NMDS if the minimum Hamming distance $d$ satisfies $d=n-k$, and {\color{red}the dual code $\C^\bot$ of $\C$ also has the same minimum distance, i.e. $d^\bot=n-k$ \cite{AMAR}}.
NMDS codes are slightly less optimal than MDS codes, while still maintaining a high level of error correction.

% Let's first recall the definition of NMDS codes and explore how they differ from MDS codes, as well as the circumstances under which TGRS codes can achieve the NMDS property. Formally,

% Next, we will introduce the equivalence conditions for NMDS codes.

In the following, we introduce a result to study NMDS codes.
\begin{lem}(\cite[Lemma 3.7]{SUIL}, \cite{DODU})\label{L5}
An $[n,k]$ linear code $\C$ over $\mathbb{F}_q$ is NMDS if and only if a generator matrix $G$ of $\C$ satisfies the following conditions:
\begin{enumerate}
\item There exists $k$ linearly dependent columns in $G$, i.e., $S(\C)\neq{0}$ and $S(\C^\perp)\neq{0}$.
\item Any $k+1$ columns of $G$ are rank of $k$, i.e., $S(\C)\leq{1}$.
\item Any $k-1$ columns of $G$ are linearly independent, i.e., $S(\C^\perp)\leq{1}$.
\end{enumerate}
\end{lem}

Now we provide a necessary and sufficient condition for the $(\mL,\mP)$-TGRS codes to be NMDS under the case that it is self-dual.

\begin{thm}\label{T2}
With the notation as in Theorem \ref{T1}, assume that the $(\mL,\mP)$-TGRS code $\C(\mL,\mP,B)$ defined by \eqref{S6} is self-dual and $B\in{\mathbb{F}_q^{k\times(n-k)}\backslash{\Omega}}$ with $\Omega$ defined by \eqref{Omega} .
% suppose that there is $B\in{\mathbb{F}_q^{k\times(n-k)}\backslash{\Omega}}$ such that $\C=\C_{n,k}(\boldsymbol{\alpha},\boldsymbol{\nu},\boldsymbol{B})$ is self-dual.
Then $\C(\mL,\mP,B)$ is NMDS if and only if for any $(k+1)$-subset $J\subseteq\{1,...,n\}$, there exists a $k$-subset $\mathcal{T}\subseteq{J}$ such that $\vert{I_k+BF_\mathcal{T}}\vert\neq0$.
\end{thm}
\begin{proof}
Let $B\in{\mathbb{F}_q^{k\times(n-k)}\backslash{\Omega}}$ and $\C(\mL,\mP,B)$ be a self-dual code, where $\Omega$ is defined by \eqref{Omega}. It then follows that $S(\C)=S(\C^\perp)\geq1$, which satisfies condition 1) of Lemma \ref{L5}. We only need to prove 2) of Lemma \ref{L5} since $\C(\mL,\mP,B)$ is self-dual. Similar to the proof of Theorem  \ref{T1}, it follows that condition 2) of Lemma \ref{L5} holds if and only if there exists a $k$-subset $\mathcal{T}\subseteq{J}$ such that $\vert{I_k+BF_\mathcal{T}}\vert\neq0$ for any $(k+1)$-subset $J\subseteq\{1,...,n\}$. This completes the proof.
\end{proof}

\begin{remark}
Note that Theorem \ref{T2} extends the result in \cite[Theorem 3.8]{SUIL} from $2$-TGRS codes to $(\mL,\mP)$-TGRS codes for the general case. A natural question is to characterize the necessary and sufficient condition for $(\mL,\mP)$-TGRS codes to be NMDS without any restrictions.
%Moreover, when $b_{k-1,0}\in \fq^*$ and $b_{i,j}=0$ for all other $(i,j)$ in the coefficient matrix $B$, the $(\mL,\mP)$-TGRS code is an NMDS code by \cite[Corollary 3.9]{SUIL}.
\end{remark}

\section{The parity check matrices of $(\mL,\mP)$-TGRS codes and the self-dual codes} \label{sect-4}
In this section, we first characterize the parity check matrices of $(\mL,\mP)$-TGRS codes for the most general case, and then investigate the self-dual codes from $(\mL,\mP)$-TGRS codes.
% Recall that the inner product of two vectors $\boldsymbol{x}=(x_1,...,x_n),\ \boldsymbol{y}=(y_1,...,y_n)\in{ \mathbb{F}_q^n}$ is defined by
% \begin{eqnarray*}
% \boldsymbol{x}\cdot\boldsymbol{y}=\sum_{i=1}^nx_iy_i.
% \end{eqnarray*}
% Then, the dual code of a linear code $\C$ of length $n$ over $\mathbb{F}_q$ is define as
% \begin{eqnarray*}
% \C^{\perp}=\{\boldsymbol{x}\mid\boldsymbol{x}\cdot\boldsymbol{y}=0,\ \forall{\boldsymbol{y}\in{\C}}\}.
% \end{eqnarray*}

\subsection{The parity check matrices of $(\mL,\mP)$-TGRS codes}
The parity check matrix of a linear code is essential since its dual code can be completely determined by its parity check matrix. In this subsection, we determine the parity check matrix of the $(\mL,\mP)$-TGRS code. 

% \begin{lem}(\cite[Theorem 4.2]{SUIS})\label{L6}
% Let $\boldsymbol{\alpha}=(a_1,a_2,...,a_n)\in\mathbb{F}_{q}^n$ with distinct $a_1,...,a_n$, $\boldsymbol{\nu}=(v_1,...,v_n)\in(\mathbb{F}_{q}^*)^n$ and $B=[b_{i,j}] \in \fq^{k\times{(n-k)}}$ and $F(\mL,\mP,B)$ be defined as in \eqref{S8-B} and \eqref{S5} respectively. Let $u_i=\prod_{j=1,j\neq{i}}^n(a_i-a_j)^{-1}$ for $1\leq{i}\leq{n}$ and $\prod_{i=1}^n(x-a_i)=\sum_{j=0}^nc_jx^{n-j}$, which defines $u_i$'s and $c_j$'s for given $a_i$'s. Then, the code $\C_{n,k}(\boldsymbol{\alpha},\boldsymbol{\nu},\boldsymbol{B})$ has parity check matrix
% \begin{equation*}
% H=[I_{n-k}\mid{-J_{n-k}B^{T}J_k}]MU,
% \end{equation*}
% where $I_{n-k}$ is is the ${(n-k)}\times{(n-k)}$ identity matrix, $J_k,\ M\ \mbox{and}\ U$ are defined as follows
% \begin{equation*}
% J_{n-k}=\begin{pmatrix}
% 0&\cdots&0&1\\
% 0&\cdots&1&0\\
% \vdots&\ddots&\vdots&\vdots\\
% 1&\cdots&0&0
% \end{pmatrix}_{(n-k)\times(n-k)},
% M=\begin{pmatrix}
% 1&\cdots&1\\
% \sum\limits_{t=0}^1c_{1-t}a_1^t&\cdots&\sum\limits_{t=0}^1c_{1-t}a_n^t\\
% \vdots&\ddots&\vdots\\
% \sum\limits_{t=0}^{n-1}c_{n-1-t}a_1^t&\cdots&\sum\limits_{t=0}^{n-1}c_{n-1-t}a_n^t
% \end{pmatrix}_{n\times{n}},
% U=\begin{pmatrix}
% \frac{u_1}{v_1}&&&\\
%   &\frac{u_2}{v_2}&&\\
%   &&\ddots&\\
%   &&&\frac{u_n}{v_n}\\
% \end{pmatrix}_{n\times{n}}.
% \end{equation*}
% \end{lem}

\begin{thm}\label{T3}
Let $\boldsymbol{\alpha}=(a_1,a_2,...,a_n)\in\mathbb{F}_{q}^n$ with distinct $a_1,...,a_n$, $\boldsymbol{\nu}=(v_1,...,v_n)\in(\mathbb{F}_{q}^*)^n$, and $B=[b_{i,j}] \in \fq^{k\times{(n-k)}}$ and $F(\mL,\mP,B)$ be defined as in \eqref{S8-B} and \eqref{S5} respectively. Define $u_i=\prod_{j=1,j\neq{i}}^n(a_i-a_j)^{-1}$ for $1\leq{i}\leq{n}$ and $\prod_{i=1}^n(x-a_i)=\sum_{j=0}^nc_jx^{n-j}$, which defines $u_i$'s and $c_j$'s for given $a_i$'s. Then the $(\mL,\mP)$-TGRS code $\C(\mL,\mP,B)$ defined by \eqref{S6} has parity check matrix as follows
\begin{equation}\label{S16}
H=\begin{pmatrix}
\cdots&\frac{u_j}{v_j}[1-\sum\limits_{i=0}^{k-1}b_{i,n-k-1}\sum\limits_{t=0}^{n-1-i}c_{n-1-i-t}a_j^t]&\cdots\\
\cdots&\frac{u_j}{v_j}[\sum\limits_{t=0}^1c_{1-t}a_j^t-\sum\limits_{i=0}^{k-1}b_{i,n-k-2}\sum\limits_{t=0}^{n-1-i}c_{n-1-i-t}a_j^t]&\cdots\\
\vdots&\vdots&\vdots\\
\cdots&\frac{u_j}{v_j}[\sum\limits_{t=0}^{n-k-2}c_{n-k-2-t}a_j^t-\sum\limits_{i=0}^{k-1}b_{i,1}\sum\limits_{t=0}^{n-1-i}c_{n-1-i-t}a_j^t]&\cdots\\
\cdots&\frac{u_j}{v_j}[\sum\limits_{t=0}^{n-k-1}c_{n-k-1-t}a_j^t-\sum\limits_{i=0}^{k-1}b_{i,0}\sum\limits_{t=0}^{n-1-i}c_{n-1-i-t}a_j^t]&\cdots\\
\end{pmatrix}.
\end{equation}
\end{thm}

\begin{proof}
We first prove ${\rm rank}(H)=n-k$ for $H$ defined as in \eqref{S16}. Denote $\boldsymbol{\alpha}^i=(a_1^i,...,a_n^i)$, $\frac{\boldsymbol{u}}{\boldsymbol{\nu}}=(\frac{u_1}{v_1},...,\frac{u_n}{v_n})$ and $\frac{\boldsymbol{u}}{\boldsymbol{\nu}}\ast\boldsymbol{\alpha}^i=(\frac{u_1}{v_1}a_1^i,...,\frac{u_n}{v_n}a_n^i)$ for $\boldsymbol{u}=(u_1,\ldots,u_n)$. Then by \eqref{S8-GM} and \eqref{S9-G} the generator matrix of $\C(\mL,\mP,B)$ can be rewritten as
\begin{equation}\label{S17}
G=\begin{pmatrix}
\boldsymbol{\nu}\ast(\boldsymbol{\alpha}^0+\sum\limits_{j=0}^{n-k-1}b_{0,j}\boldsymbol{\alpha^{k+j}})\\
\boldsymbol{\nu}\ast(\boldsymbol{\alpha}^1+\sum\limits_{j=0}^{n-k-1}b_{1,j}\boldsymbol{\alpha^{k+j}})\\
\vdots\\
\boldsymbol{\nu}\ast(\boldsymbol{\alpha}^{k-1}+\sum\limits_{j=0}^{n-k-1}b_{k-1,j}\boldsymbol{\alpha^{k+j}})
\end{pmatrix}
=[I_k\mid{B}]V_nV_0,
\end{equation}
where $I_k$ is the $k\times{k}$ identity matrix, $V_n=(\boldsymbol{\alpha}^0,...,\boldsymbol{\alpha}^{n-1})^{\rm T}$ is the $n\times{n}$ Vandermonde matrix, and $V_0$ is a diagonal matrix with elements $v_1,v_2,\ldots,v_n$, given by
\begin{equation*}
V_n=\begin{pmatrix}
1&1&\cdots&1\\
a_1&a_2&\cdots&a_n\\
\vdots&\vdots&\ddots&\vdots\\
a_1^{n-1}&a_2^{n-1}&\cdots&a_n^{n-1}\\
\end{pmatrix},
V_0=\begin{pmatrix}
v_1&&&\\
&v_2&&\\
&&\ddots&\\
&&&v_n
\end{pmatrix}.
\end{equation*}

It is clear that $H$ in (\ref{S16}) can be expressed as
\begin{equation} \label{S18}
H=\begin{pmatrix}
\frac{\boldsymbol{u}}{\boldsymbol{\nu}}\ast[\boldsymbol{{\alpha}^0}-\sum\limits_{i=0}^{k-1}b_{i,n-k-1}\sum\limits_{t=0}^{n-1-i}c_{n-1-i-t}\boldsymbol{{\alpha}^t}]\\
\frac{\boldsymbol{u}}{\boldsymbol{\nu}}\ast[\sum\limits_{t=0}^1c_{1-t}\boldsymbol{{\alpha}^t}-\sum\limits_{i=0}^{k-1}b_{i,n-k-2}\sum\limits_{t=0}^{n-1-i}c_{n-1-i-t}\boldsymbol{{\alpha}^t}]\\
\vdots\\
\frac{\boldsymbol{u}}{\boldsymbol{\nu}}\ast[\sum\limits_{t=0}^{n-k-2}c_{n-k-2-t}\boldsymbol{{\alpha}^t}-\sum\limits_{i=0}^{k-1}b_{i,1}\sum\limits_{t=0}^{n-1-i}c_{n-1-i-t}\boldsymbol{{\alpha}^t}]\\
\frac{\boldsymbol{u}}{\boldsymbol{\nu}}\ast[\sum\limits_{t=0}^{n-k-1}c_{n-k-1-t}\boldsymbol{{\alpha}^t}-\sum\limits_{i=0}^{k-1}b_{i,0}\sum\limits_{t=0}^{n-1-i}c_{n-1-i-t}\boldsymbol{{\alpha}^t}]
\end{pmatrix}.
\end{equation}
Further, it can be written as
\begin{equation*}
H=[-J_{n-k}B^{\rm T}\mid{J_{n-k}}]CV_nU,
\end{equation*}
where $J_{n-k} \in \fq^{(n-k) \times (n-k)}$, $C \in \fq^{n\times n}$ and $U\in\fq^{n\times n}$ are defined as follows:
\[
\begin{aligned}
%\begin{eqnarray*}
&J_{n-k}=\begin{pmatrix}
0&\cdots&0&1\\
0&\cdots&1&0\\
\vdots&\ddots&\vdots&\vdots\\
1&\cdots&0&0
\end{pmatrix},
C=\begin{pmatrix}
c_{n-1}&c_{n-2}&\cdots&1\\
c_{n-2}&c_{n-3}&\cdots&0\\
\vdots&\vdots&&\vdots\\
c_1&1&\cdots&0\\
1&0&\cdots&0\\
\end{pmatrix},
%\end{eqnarray*}
%\begin{eqnarray*}
&U=\begin{pmatrix}
\frac{u_1}{v_1}&&&\\
  &\frac{u_2}{v_2}&&\\
  &&\ddots&\\
  &&&\frac{u_n}{v_n}\\
\end{pmatrix}.
%\end{eqnarray*}
\end{aligned}
\]
Observe that $J_{n-k}$, $C$ and $U$ are invertible. We then conclude that
\begin{equation*}
{\rm rank}(H)={\rm rank}([-J_{n-k}B^{\rm T}\mid{J_{n-k}}])=n-k.
\end{equation*}
This proves ${\rm rank}(H)=n-k$.

Next we prove that $GH^{\rm T}=\boldsymbol{0}$. It gives
\begin{equation*}
GH^{\rm T}=[I_k\mid{B}]V_nV_0([-J_{n-k}B^{\rm T}\mid{J_{n-k}}]CV_nU)^{\rm T},
\end{equation*}
which can be written as
\begin{equation} \label{GH}
GH^{\rm T}=[I_k\mid{B}]V_nV_0UV_n^{\rm T}C^{\rm T}\left[\begin{array}{c}-BJ_{n-k}\\ \hline J_{n-k}\end{array}\right].
\end{equation}
A direct computation gives
$$
\begin{aligned}
V_nV_0UV_n^{\rm T}&=\begin{pmatrix}\boldsymbol{\alpha^0}\\\boldsymbol{\alpha}\\\vdots\\\boldsymbol{\alpha}^{n-1}\end{pmatrix}
((\boldsymbol{u}\ast\boldsymbol{\alpha}^0)^{\rm T},...,(\boldsymbol{u}\ast\boldsymbol{\alpha}^{n-1})^{\rm T})
=\begin{pmatrix}
\sum\limits_{i=1}^nu_ia_i^0&\sum\limits_{i=1}^nu_ia_i^1&\cdots&\sum\limits_{i=1}^nu_ia_i^{n-1}\\
\sum\limits_{i=1}^nu_ia_i^1&\sum\limits_{i=1}^nu_ia_i^2&\cdots&\sum\limits_{i=1}^nu_ia_i^n\\
\vdots&\vdots&\ddots&\vdots\\
\sum\limits_{i=1}^nu_ia_i^{n-1}&\sum\limits_{i=1}^nu_ia_i^n&\cdots&\sum\limits_{i=1}^nu_ia_i^{2n-2}
\end{pmatrix}.
\end{aligned}
$$
It follows from the proof of \cite[Theorem 2.2]{HUAN} that
\begin{equation*}
  \left\{\begin{array}{ll}
    \sum\limits_{t=1}^nu_ta_t^i=0,    &   \mbox{ if } 0\leq{i}\leq{n-2};\\
    \sum\limits_{t=1}^nu_ta_t^i=1,    &   \mbox{ if } i=n-1.\\
\end{array}\right.
\end{equation*}
By using Lemma \ref{L4}, then $a_i^{n+t}=\sum\limits_{s=0}^{n-1}f_{t,s}a_i^s$. For $0\leq{t}\leq{n-2}$, we have
\begin{equation*}
\sum\limits_{i=1}^nu_ia_i^{n+t}=\sum\limits_{i=1}^nu_i\sum\limits_{s=0}^{n-1}f_{t,s}a_i^s=\sum\limits_{s=0}^{n-1}f_{t,s}\sum\limits_{i=1}^nu_ia_i^s=f_{t,n-1},\\
\end{equation*}
and by \eqref{S11} it gives
\begin{equation*}
f_{t,n-1}=-\sum_{i=0}^{\min\{t,n-1\}}c_{i+1}e_{t-i}=-\sum\limits_{i=0}^{t}c_{i+1}e_{t-i}=-\sum\limits_{j=0}^{t}c_{t+1-j}e_{j}=e_{t+1}.
\end{equation*}
It leads to
$$
\begin{aligned}
V_nV_0UV_n^T&=\begin{pmatrix}
\sum\limits_{i=1}^nu_ia_i^0&\sum\limits_{i=1}^nu_ia_i^1&\cdots&\sum\limits_{i=1}^nu_ia_i^{n-1}\\
\sum\limits_{i=1}^nu_ia_i^1&\sum\limits_{i=1}^nu_ia_i^2&\cdots&\sum\limits_{i=1}^nu_ia_i^n\\
\vdots&\vdots&\ddots&\vdots\\
\sum\limits_{i=1}^nu_ia_i^{n-1}&\sum\limits_{i=1}^nu_ia_i^n&\cdots&\sum\limits_{i=1}^nu_ia_i^{2n-2}
\end{pmatrix}
=\begin{pmatrix}
0&\cdots&0&e_0\\
0&\cdots&e_0&e_1\\
\vdots&&\vdots&\vdots\\
0&\cdots&e_{n-3}&e_{n-2}\\
e_0&\cdots&e_{n-2}&e_{n-1}
\end{pmatrix}.
\end{aligned}
$$
By using Lemma \ref{L3}, it shows that
\begin{equation*}
\begin{pmatrix}
c_0&&&&\\
c_1&c_0&&&\\
c_2&c_1&c_0&&\\
\vdots&\vdots&\ddots&\ddots\\
c_{n-1}&c_{n-2}&\cdots&c_1&c_0\\
\end{pmatrix}
\begin{pmatrix}
e_0&&&&\\
e_1&e_0&&&\\
e_2&e_1&e_0&&\\
\vdots&\vdots&\ddots&\ddots\\
e_{n-1}&e_{n-2}&\cdots&e_1&e_0\\
\end{pmatrix}
=I_n.
\end{equation*}
Then we have
\begin{equation*}
V_nV_0UV_n^{\rm T}C^{\rm T}=
\begin{pmatrix}
0&\cdots&0&e_0\\
0&\cdots&e_0&e_1\\
\vdots&&\vdots&\vdots\\
0&\cdots&e_{n-3}&e_{n-2}\\
e_0&\cdots&e_{n-2}&e_{n-1}\\
\end{pmatrix}
\begin{pmatrix}
c_{n-1}&c_{n-2}&\cdots&c_0\\
c_{n-2}&c_{n-3}&\cdots&0\\
\vdots&\vdots&&\vdots\\
c_1&c_0&\cdots&0\\
c_0&0&\cdots&0\\
\end{pmatrix}=I_n.
\end{equation*}
This together with \eqref{GH} gives
$$
\begin{aligned}
GH^{\rm T}&=[I_k\mid{B}]V_nV_0UV_n^{\rm T}C^{\rm T}\left[\begin{array}{c}-BJ_{n-k}\\ \hline J_{n-k}\end{array}\right]\\
&=[I_k\mid{B}]\left[\begin{array}{c}-BJ_{n-k}\\ \hline J_{n-k}\end{array}\right]\\
&=\boldsymbol{0}.
\end{aligned}
$$
This completes the proof.
\end{proof}

\begin{remark}
In Theorem \ref{T3}, we present an explicit characterization of the parity check matrices of the $(\mL,\mP)$-TGRS codes for the most general case by using the formula given in \cite[Theorem 4.2]{SUIS}. Moreover, Theorem \ref{T3} extend the results in \cite[Theorem 7]{CHEN} in which the TGRS codes with at most $\ell$ positions of $B$ being nonzero are considered.
\end{remark}

\subsection{The self-dual $(\mL,\mP)$-TGRS codes}
In this section, we study the self-dual codes from $(\mL,\mP)$-TGRS codes.
% Recall that an $[n,k]$ linear code $\C$ over $\mathbb{F}_q$ is called a self-dual code if $\C=\C^\perp$.
If $\C$ has a generator matrix $G$ and a parity check matrix $H$, then $\C={\rm span}_{\fq}(G)$ and $\C^\perp={\rm span}_{\fq}(H)$. Therefore, $\C$ is self-dual if and only if ${\rm span}_{\mathbb{F}_q}(G)={\rm span}_{\mathbb{F}_q}(H)$.
In the following, we always assume that $n=2k$.

\begin{thm}\label{T4}
Let $\boldsymbol{\alpha}=(a_1,a_2,...,a_n)\in\mathbb{F}_{q}^n$ with distinct $a_1,...,a_n$, $\boldsymbol{\nu}=(v_1,...,v_n)\in(\mathbb{F}_{q}^*)^n$, and $B=[b_{i,j}] \in \fq^{k\times{(n-k)}}$ and $F(\mL,\mP,B)$ be defined as in \eqref{S8-B} and \eqref{S5} respectively. Define $u_i=\prod_{j=1,j\neq{i}}^n(a_i-a_j)^{-1}$ for $1\leq{i}\leq{n}$ and $\prod_{i=1}^n(x-a_i)=\sum_{j=0}^nc_jx^{n-j}$, which defines $u_i$'s and $c_j$'s for given $a_i$'s. Assume that $n=2k$.
Then the $(\mL,\mP)$-TGRS code $\C(\mL,\mP,B)$ defined by \eqref{S6} is self-dual if the following two conditions hold:
\begin{enumerate}
\item There exists a $\lambda\in\mathbb{F}_{q}^*$ such that $v_i^2=\lambda{u_i}$ for all $1\leq{i}\leq{n}$;
\item $B^{\rm T}DB=NB+B^{\rm T}N$, where
$
D=\begin{pmatrix}
c_{n-1}&\cdots&c_k\\
\vdots&&\vdots\\
c_k&\cdots&c_1\\
\end{pmatrix}
$ 
and 
$N=\begin{pmatrix}
c_{k-1}&c_{k-2}&\cdots&1\\
c_{k-2}&c_{k-3}&\cdots&0\\
\vdots&\vdots&&\vdots\\
c_1&1&\cdots&0\\
1&0&\cdots&0
\end{pmatrix}.
$
\end{enumerate}
\end{thm}
\begin{proof}
Recall the generator matrix $G$ and the parity check matrix $H$ of $\C(\mL,\mP,B)$ given as in (\ref{S17}) and (\ref{S18}) respectively. Let $G=[g_0,\ldots,g_{k-1}]^{\rm T}$ and $H=[h_0,\ldots,h_{n-k-1}]^{\rm T}$, where $g_i$ for $0\leq i \leq k-1$ (resp. $h_j$ for $0\leq j \leq n-k-1$) denotes the $(i+1)$-th (resp. $j+1$) row of $G$ (resp. $H$). The code $\C(\mL,\mP,B)$ is self-dual if and only if the sets $\{g_0,...,g_{k-1}\}$ and $\{h_0,...,h_{k-1}\}$ are linearly related to each other.

From the proof of Theorem \ref{T3}, the matrix G in (\ref{S17}) can be rewritten as
\begin{equation*}
G=[I_k\mid{B}]\begin{pmatrix}
\boldsymbol{\nu}\ast\boldsymbol{\alpha}^0\\
\boldsymbol{\nu}\ast\boldsymbol{\alpha}^1\\
\vdots\\
\boldsymbol{\nu}\ast\boldsymbol{\alpha}^{n-1}\\
\end{pmatrix}.
\end{equation*}
Similarly, the parity check matrix $H$ in (\ref{S18}) can be rewritten as
\begin{equation*}
H=[-J_kB^{\rm T}\mid{J_k}]C\begin{pmatrix}
\frac{\boldsymbol{u}}{\boldsymbol{\nu}}\ast\boldsymbol{\alpha}^0\\
\frac{\boldsymbol{u}}{\boldsymbol{\nu}}\ast\boldsymbol{\alpha}^1\\
\vdots\\
\frac{\boldsymbol{u}}{\boldsymbol{\nu}}\ast\boldsymbol{\alpha}^{n-1}\\
\end{pmatrix},
\end{equation*}
where $J_k \in \fq^{k \times k}$ and $C\in \fq^{n\times n}$ are given by
\begin{equation*}
J_k=\begin{pmatrix}
0&\cdots&0&1\\
0&\cdots&1&0\\
\vdots&\ddots&\vdots&\vdots\\
1&\cdots&0&0
\end{pmatrix},
C=\begin{pmatrix}
c_{n-1}&c_{n-2}&\cdots&1\\
c_{n-2}&c_{n-3}&\cdots&0\\
\vdots&\vdots&&\vdots\\
c_1&1&\cdots&0\\
1&0&\cdots&0
\end{pmatrix}.
\end{equation*}
Then $\C(\mL,\mP,B)$ is self-dual if the following two conditions hold:
\begin{enumerate}
\item [a)] There exists a $\lambda\in\mathbb{F}_{q}^*$ such that $v_i^2=\lambda{u_i}$ for all $1\leq{i}\leq{n}$;
\item [b)] There exists a nonsingular matrix $M\in \fq^{k \times k}$ such that $[I_k\mid{B}]=M[-J_kB^{\rm T}\mid{J_k}]C$.
\end{enumerate}

Next we demonstrate that the condition b) is equivalent to condition 2) in Theorem \ref{T4}. Note that the matrix $C$ can be expressed as
$\begin{pmatrix}
D&N\\
N&\boldsymbol{0}_{k\times{k}}\\
\end{pmatrix},
$
where $D\in \fq^{k \times k}$ and $N\in \fq^{k \times k}$ are given by
\begin{equation*}
D=\begin{pmatrix}
c_{n-1}&\cdots&c_k\\
\vdots&&\vdots\\
c_k&\cdots&c_1
\end{pmatrix},
N=
\begin{pmatrix}
c_{k-1}&c_{k-2}&\cdots&1\\
c_{k-2}&c_{k-3}&\cdots&0\\
\vdots&\vdots&&\vdots\\
c_1&1&\cdots&0\\
1&0&\cdots&0
\end{pmatrix}.
\end{equation*}
% and
% \begin{equation*}
% N=\begin{pmatrix}
% c_{k-1}&\cdots&c_1&1\\
% \vdots&&&\\
% c_1&1&&\\
% 1&&&\\
% \end{pmatrix}_{k\times{k}}.
% \end{equation*}
% Then it follows that
% \[[I_k\mid{B}]=M[-J_kB^T\mid{J_k}]C=M[-J_kB^TD+J_kN\mid{-J_kB^TN}],\]
% which implies that $M(-J_kB^TD+J_kN)=I_k$ and $M(-J_kB^TN)=B$.
% This gives
% \begin{equation*}
% M(-J_kB^TD+J_kN)B= M(-J_kB^TN),
% \end{equation*}
% which indicates that
% \begin{equation*}
% B^TDB=NB+B^TN.
% \end{equation*}

Suppose that $B^{\rm T}DB=NB+B^{\rm T}N$. Then $-B^{\rm T}N=(-B^{\rm T}D+N)B$, and it gives
\begin{equation*}
[-J_kB^{\rm T}\mid{J_k}]C=[-J_kB^{\rm T}D+J_kN\mid{-J_kB^{\rm T}N}]=(J_k)(-B^{\rm T}D+N)[I_k\mid{B}].
\end{equation*}
Then we have $\mbox{rank}((J_k)(-B^{\rm T}D+N))=k$ due to the fact that $\mbox{rank}([-J_kB^{\rm T}\mid{J_k}]C)=k$.
Therefore there exists a nonsingular $M=((J_k)(-B^{\rm T}D+N))^{-1}$ such that the condition b) holds.
This completes the proof.
\end{proof}

\begin{remark}
In Theorem \ref{T4}, we provide a sufficient condition for $(\mL,\mP)$-TGRS codes to be self-dual. A natural question is to characterize the necessary and sufficient condition for $(\mL,\mP)$-TGRS codes to be self-dual for the most general case.
\end{remark}

\section{The non-GRS properties of $(\mL,\mP)$-TGRS codes} \label{sect-5}
In this section, we will study the non-GRS properties of the $(\mL,\mP)$-TGRS codes for the most general case. 
It worth noting that constructing non-GRS MDS codes is an interesting research topic since most of the known MDS codes are equivalent to GRS codes.
It is shown in \cite{BERO} that many TGRS codes are non-GRS codes for certain $\mL$, $\mP$ and coefficient matrix $B$. We will use a similar approach to the one in \cite{BERO} to explore the non-GRS properties of $(\mL,\mP)$-TGRS codes.

\subsection{Inequivalence based on the Schur square}
In this subsection, we investigate the inequivalence of $(\mL,\mP)$-TGRS codes to GRS codes by using the Schur square.

The study of Schur squares plays a significant role in coding theory due to their applications \cite{COUV,CRAD,RAND}. Next we introduce the definition of Schur square of a linear code over $\fq$.
\begin{definition}(\cite{BERO})
Let $\C$ be an $[n,k]$ linear code over $\fq$. The Schur square of $\C$ is a linear codes over $\fq$ defined by
\begin{equation*}
\C^2:=\langle\{c\star c^\prime:\ c,\ c^\prime\in{\C}\}\rangle,
\end{equation*}
where $c\star c^\prime=(c_1c_1^\prime,\ldots,c_nc_n^\prime)$ denotes the Schur product of $c=(c_1,\ldots,c_{n})\in \fq^n$ and $c'=(c'_1,\ldots,c'_{n})\in \fq^n$, and $\langle S \rangle $ represent the $\fq$-subspace spanned by the set $S$ of $\fq^n$.
\end{definition}

The dimension of the Schur product of a code is an invariant up to equivalence of codes. For any linear code $\C$ over $\fq$, it satisfies the inequality \cite{MIRA} that
\begin{equation*}
\mbox{dim}(\C^2)\leq{\mbox{min}\{n,\frac{1}{2}k(k+1)\}}.
\end{equation*}
A random linear code attains this upper bound with high probability \cite{CASC}. For an MDS code $\C$, it satisfies $\mbox{dim}(\C^2)\geq \mbox{min}\{n,2k-1\}$ \cite{RAND}, and specially $\mbox{dim}(\C^2)=\mbox{min}\{n,2k-1\}$ for a GRS code $\C$. 

% We start with a generic lower bound on the Schur square dimension of an evaluation code.
Next we give a generic lower bound on the dimension of Schur square of the evaluation code, which is generated by using the evaluation map $ev_{\alpha,v}$. Before this, we introduce the following definition.

\begin{definition}(\cite{BERO}) \label{def-DP}
Let $\Delta$ be an $\fq$-subspace of $\mathbb{F}_{q}[x]_{<n}$, and $\boldsymbol{\alpha}=(a_1,\ldots,a_n)\in\mathbb{F}_{q}^n$ with $a_1,\ldots,a_n$ distinct. Define the polynomial sets $D(\Delta)_{<n}$ and $\overline{D}(\Delta,\boldsymbol{\alpha})$ as follows:
\begin{equation*}
D(\Delta)_{<n}=\{\deg(f(x)g(x)): f(x), g(x)\in{\Delta},\deg(f(x)g(x))<n\}
\end{equation*}
and
\begin{equation*}
\overline{D}(\Delta,\boldsymbol{\alpha})=\{\deg(\overline{f(x)g(x)}):\ f(x),\ g(x)\in \Delta\},
\end{equation*}
where $\overline{f(x)}= f(x) \mbox{ \rm mod } \prod_{i=1}^n(x-a_i)$ for $f(x)\in{\mathbb{F}_{q}[x]}$.
\end{definition}

In the following, we directly extend the result in \cite[Lemma 9]{BERO} from $\boldsymbol{\nu}=(1,\ldots,1)\in(\fq)^n$ to any $\boldsymbol{\nu}=(v_1,\ldots,v_n)\in(\mathbb{F}_{q}^*)^n$ and we omit the proof since it can be similarly proved.

\begin{lem}\label{L7}
Let $\boldsymbol{\alpha}=(a_1,\ldots,a_n)\in\mathbb{F}_{q}^n$ with $a_1,\ldots,a_n$ distinct, and $\boldsymbol{\nu}=(v_1,...,v_n)\in(\mathbb{F}_{q}^*)^n$. Let $\Delta$, $D(\Delta)_{<n}$ and $\overline{D}(\Delta,\boldsymbol{\alpha})$ be defined as in Definition \ref{def-DP}, and $\C=ev_{\boldsymbol{\alpha},\boldsymbol{\nu}}(\Delta)$ be the evaluation code of $\Delta$. Then
\begin{equation*}
\C^2=ev_{\boldsymbol{\alpha},\boldsymbol{\nu}^2}(\langle f(x)g(x):\ f(x),\ g(x)\in \Delta\rangle)
\end{equation*}
and
\begin{equation*}
\dim(C^2)\geq \vert{\overline{D}(\Delta,\boldsymbol{\alpha})}\vert\geq\vert{D(\Delta)_{<n}}\vert.
\end{equation*}
\end{lem}
% \begin{proof}
% The proof follows similarly to \cite[Lemma 9]{BERO} and the result can be directly obtained in the same way.
% \end{proof}

In the following theorem, we study the non-GRS property of $(\mL,\mP)$-TGRS codes with respect to a special form of coefficient matrix $B$.
% will restrict the form of the coefficient matrix as defined in Definition \ref{D2}.  investigate the inequivalence of $(\mL,\mP)$-TGRS codes to GRS codes by using Schur squares.
% This will result in the corresponding TGRS codes containing many low-rate (smaller than $\frac{1}{2}$) non-GRS codes.

\begin{thm}\label{T5}
Let $n>2k$ and the $(\mL,\mP)$-TGRS codes $\C(\mL,\mP,B)$ be defined as in \eqref{S6}. Let the coefficient matrix $B$ be given by
\begin{equation}\label{B-sepcial}
B=\begin{pmatrix}
\boldsymbol{0}_{(k-\ell)\times{\ell}}&\boldsymbol{0}_{(k-\ell)\times{(n-k-\ell)}}\\
\boldsymbol{A}_{\ell\times{\ell}}&\boldsymbol{0}_{\ell\times{(n-k-\ell)}}\\
\end{pmatrix},
\end{equation}
where $\ell<\min\left\{k,n-2k+1\right\}$ and
\begin{equation*}
A=\begin{pmatrix}
b_{k-\ell,0}&0&\cdots&0\\
b_{k-\ell+1,0}&b_{k-\ell+1,1}&\cdots&0\\
\vdots&\vdots&\ddots&\vdots\\
b_{k-1,0}&b_{k-1,1}&\cdots&b_{k-1,\ell-1}\\
\end{pmatrix}.
\end{equation*}
Then the dimension of the Schur square of $\C(\mL,\mP,B)$ is $\dim(\C(\mL,\mP,B)^2)\geq 2k$, and $\C(\mL,\mP,B)$ is non-GRS. Moreover, $\C(\mL,\mP,B)$ is a non-GRS MDS code if $B\in \Omega$ with $\Omega$ defined as in \eqref{Omega}.
\end{thm}
\begin{proof}
By Lemma \ref{L1}, the set $F(\mL,\mP,B)$ of twisted polynomials of $\C(\mL,\mP,B)$ for the given $B$ as in \eqref{B-sepcial} has a basis $\{g_i(x):0\leq i \leq k-1\}$, where
\begin{equation*}
  \left\{\begin{array}{ll}
    g_i(x)=x^i,    &   \mbox{ if } 0\leq{i}\leq{k-\ell-1};\\
    g_i(x)=x^i+\sum_{j=0}^{i-k+\ell}b_{ij}x^{k+j},    &   \mbox{ if } k-\ell\leq{i}\leq{k-1}.
\end{array}\right.
\end{equation*}
Notice that $\{\deg(g_i(x)): 0\leq i \leq k-1\}$ is given by
\begin{equation*}
  S(B)=\{0,1,...,k-\ell-1,k,k+1,...,k+\ell-1\}.
\end{equation*}
Further, we define the set $\Upsilon =\{f(x)g(x):f(x),g(x)\in{F(\mL,\mP,B)}\}$. Notice that $\Upsilon$ must contain polynomials of degree $i$ for $i\in T_1\cup T_2 \cup T_3$, where
$$T_1:=\{0,1,...,2k-2\ell-2\},T_2:=\{2k-\ell-2,2k-\ell-1,...,2k-2\},T_3:=\{2k,2k+1,...,2k+2\ell-2\}.$$
% \begin{equation*}
% \{0,1,2,...,2k-2\ell-2\},\ \{2k-\ell-2,2k-\ell-1,2k-\ell,...,2k-2\},\ \{2k,2k+1,...,2k+2\ell-2\}.
% \end{equation*}
Observe that $|T_1|+|T_2|=2k-\ell$ and $i<n$ for $i\in T_1\cup T_2$. Moreover, due to $2k+\ell-1<n$, there are at least $\ell$ elements $i\in T_3$ such that $i<n$. Then we conclude that there are at least $2k$ polynomials of distinct degrees less than $n$ in the set $\Upsilon$. This together with Lemma \ref{L7} gives that $\dim(\C(\mL,\mP,B)^2)\geq|{D(F(\mL,\mP,B))_{<n}}| \geq |\Upsilon| \geq 2k$, where $D(\cdot)_{<n}$ is defined as in Definition \ref{def-DP}. Recall that the dimension of the Schur square of a GRS code is $2k-1$ due to $n>2k$. Therefore, $\C(\mL,\mP,B)$ is non-GRS. This completes the proof.
\end{proof}

\begin{remark}
By Theorems \ref{T1} and \ref{T5}, non-GRS MDS codes can be derived from the $(\mL,\mP)$-TGRS codes for the coefficient matrix $B$ of the form \eqref{B-sepcial}.
\end{remark}

\begin{example}
Let $n=8$, $k=3$, $q=17$, $\alpha=(1,2,3,4,5,6,7,8)\in{\mathbb{F}_{17}^8}$, $\boldsymbol{\nu}=(1,\ldots,1)$ and $B$ be of the form
\begin{equation*}
\begin{pmatrix}
0&0&0&0&0\\
b_{1,0}&0&0&0&0\\
b_{2,0}&b_{2,1}&0&0&0
\end{pmatrix}.
\end{equation*}
Then $\C(\mL,\mP,B)$ is non-GRS by Theorem \ref{T5}. Magma experiments shows that $\C(\mL,\mP,B)$ is an $[8,3,6]$ non-GRS MDS code if and only if $B\in \Xi$, where $|\Xi|=76$  and $\Xi$ is given by
\begin{equation*}
\Xi=\left\{\begin{pmatrix}
0&0&0&0&0\\
12&0&0&0&0\\
1&0&0&0&0
\end{pmatrix},\begin{pmatrix}
0&0&0&0&0\\
15&0&0&0&0\\
14&9&0&0&0
\end{pmatrix},\begin{pmatrix}
0&0&0&0&0\\
13&0&0&0&0\\
8&13&0&0&0
\end{pmatrix},\begin{pmatrix}
0&0&0&0&0\\
3&0&0&0&0\\
10&0&0&0&0
\end{pmatrix},
%\begin{pmatrix}
%0&0&0&0&0\\
%16&0&0&0&0\\
%8&3&0&0&0\\
%\end{pmatrix},
\ldots\right\}.
\end{equation*}
\end{example}

The following result can be derived directly from Theorem \ref{T5}, which is a special case of Theorem \ref{T5}. Note that this type of TGRS codes was first proposed by Gu et al. \cite{GU}, while the non-GRS property of the codes has not been investigated.

\begin{cor}\label{Cor6}
Let $n >2k$ and $\C(\mL,\mP,B)$ be defined as in \eqref{S6}. Let the coefficient matrix $B$ be given by
\begin{equation*}
B=\begin{pmatrix}
\boldsymbol{0}_{(k-\ell)\times{\ell}}&\boldsymbol{0}_{(k-\ell)\times{(n-k-\ell)}}\\
\boldsymbol{E}_{\ell\times{\ell}}&\boldsymbol{0}_{\ell\times{(n-k-\ell)}}
\end{pmatrix},
\end{equation*}
where
\begin{equation*}
E=\begin{pmatrix}
b_{k-\ell,0}&&&\\
&b_{k-\ell+1,1}&&\\
&&\ddots&\\
&&&b_{k-\ell,\ell-1}
\end{pmatrix},
\end{equation*}
and $\ell<\min\{k,n-2k+1\}$. Then $\C(\mL,\mP,B)$ is non-GRS. Moreover, $\C(\mL,\mP,B)$ is a non-GRS MDS code if $B\in \Omega$ with $\Omega$ defined as in \eqref{Omega}.
\end{cor}

%\begin{example}
%Let $n=11$, $k=4$, $q=17$, $\alpha=(1,2,3,4,5,6,7,8,9,10,11)\in{\mathbb{F}_{17}^{11}}$ and $B$ be of the form
%\begin{equation*}
%\begin{pmatrix}
%0&0&0&0&0&0&0\\
%b_{1,0}&0&0&0&0&0&0\\
%0&b_{2,1}&0&0&0&0&0\\
%0&0&b_{3,2}&0&0&0&0\\
%\end{pmatrix}
%\end{equation*}
%Then $\C(\mL,\mP,B)$ is non-GRS by Corollary \ref{Cor6}.  Magma experiments shows that $\C(\mL,\mP,B)$ is an $[11,4,8]$ MDS no-GRS code if and only if $B\in \Xi$, where $|\Xi|=1$  and $\Xi$ is given by
%\begin{equation*}
%\Xi=\left\{
%\begin{pmatrix}
%0&0&0&0&0&0&0\\
%11&0&0&0&0&0&0\\
%0&4&0&0&0&0&0\\
%0&0&8&0&0&0&0\\
%\end{pmatrix}\right\}.
%\end{equation*}
%\end{example}
\subsection{A combinatorial inequivalence argument}
In this subsection, we first present some combinatorial results to studying the non-GRS property of the $(\mL,\mP)$-TGRS codes. The following result gives a well-known characterization of GRS codes.

\begin{lem}(\cite{ROTL,ROTS})\label{L8}
Let $\C$ be an $[n,k]$ linear code with a generator matrix of the form $G=[I_k|M]$, where $M=[M_{i,j}]\in \fq^{k \times (n-k)}$ and $M_{i,j}$'s are entries of $M$. Let $M'=[M'_{i,j}]\in \fq^{k\times(n-k)} $ with $M'_{i,j}=M_{ij}^{-1}$. Then $\C$ is a $GRS$ code if and only if the following conditions hold:
\begin{enumerate}
\item all entries of $M$ are non-zero;
\item all $2\times2$ minors of $M^\prime$ are non-zero; and
\item all $3\times3$ minors of $M^\prime$ are zero.
\end{enumerate}
\end{lem}

Note that an MDS code can be characterized by conditions 1) and 2). The crucial difference between a GRS code and a non-GRS MDS code lies on the condition 3). Moreover, it's known that when min$\{k,n-k\}<3$, an $[n,k]$ MDS code is always a GRS code.

Recall that the $(\mL,\mP)$-TGRS code defined as in \eqref{S6} is MDS if and only if $B\in \Omega$ by Theorem \ref{T1}, where $\Omega$ is defined as in \eqref{Omega}.
The coefficient matrix $B\in \fq^{k \times (n-k)}$ is given by
\begin{equation*}
  B=\begin{pmatrix}
    b_{0,0}&b_{0,1}&\dots&b_{0,n-k-1}\\
    b_{1,0}&b_{1,1}&\dots&b_{1,n-k-1}\\
    \vdots&\vdots&\ddots&\vdots\\
    b_{k-1,0}&b_{k-1,1}&\dots&b_{k-1,n-k-1}\\
  \end{pmatrix},
\end{equation*}
where $b_{i,j}\in \fq$ for $0\leq i\leq k-1$ and $0\leq j \leq n-k-1$.

Notably, Beelen et al. \cite{BERO} investigated the non-GRS property of a special type of MDS $(\mL,\mP)$-TGRS codes. When $\mL=\{t_1,t_2,...,t_{\ell}\}$ and $\mP=\{h_1,...,h_\ell\}$ with $\ell\leq \mbox{min}\{k,n-k\}$ and the coefficient matrix $B$ satisfies that $b_{h_i,t_j}\in\fq$ for $i=j$ and $b_{h_i,t_j}=0$ otherwise (which implies that at most $\ell$ positions of $B$ are nonzero), the $(\mL,\mP)$-TGRS codes are reduced to the TGRS codes studied in \cite{BERO}. We use the same technique to investigate the non-GRS property of $(\mL,\mP)$-TGRS codes for the most general case.

%{\color{red}When $\mL=\{t_1,t_2,...,t_{\ell}\}$ and $\mP=\{h_1,...,h_\ell\}$ with $\ell\leq \mbox{min}\{k,n-k\}$, and the coefficient matrix $B$ satisfying that $b_{h_i,t_j}=0\ \mbox{for}\ i\neq j$ and $b_{h_i,t_j}\neq0$ otherwise. Then $F_{n,k}(I,L,S)=\{\sum_{i=0}^{k-1}f_ix^i+\sum_{j=1}^\ell f_{h_j}b_{h_j,t_j}x^{k+t_j}\}$. The code $\C_{n,k}(\boldsymbol{\alpha},\textbf{1},I,L,S)$ obtained in this way is the TRS code studied in \cite{BEBO}. When $\ell$, $L$ and $I$ take specific values, we can further obtain the TGRS codes studied in \cite{CHEN,GU,SINC,SUIL}.}

For the multi-variable polynomial $\Gamma \in\fq[x_1,\ldots,x_{k(n-k)}]$, we say that $B\in \fq^{k \times (n-k)}$ is a zero of the polynomial $\Gamma$ (i.e., $\Gamma(B)=0$) if 
\begin{equation*}\label{gamma}
  \Gamma(b_{0,0},b_{0,1},\ldots,b_{0,n-k-1},b_{1,0},\ldots, b_{k-1,0},\ldots, b_{k-1,n-k-1})=0.
\end{equation*}

% for all $B=\begin{pmatrix}
%             b_{0,0}&\dots&b_{0,n-k-1}\\
%             \vdots&&\vdots\\
%             b_{k-1,0}&\dots&b_{k-1,n-k-1}\\
%             \end{pmatrix}\in{\mathcal{B}}$.

We provide the following results on $(\mL,\mP)$-TGRS codes, where the techniques in \cite{BERO} are useful in the proofs.

\begin{lem}\label{L9}
Let $\C(\mL,\mP,B)$ be an $[n,k]$ code defined by \eqref{S6} with $B\in \Omega$, where $\Omega$ represents the set of $B$'s such that $\C(\mL,\mP,B)$ is MDS and it is given as in \eqref{Omega}.
Let $G^{(sys,B)}=[I_k|M^{(B)}]$ be the systematic generator matrix of $\C(\mL,\mP,B)$. Then the entries of $\boldsymbol{M}^{(B)}\in{\fq^{k\times{(n-k)}}}$ can be written as
\begin{equation}\label{S19}
M_{i,j}^{(B)}=\frac{p^{(i,j)}(b_{0,0},b_{0,1},\ldots,b_{0,n-k-1},b_{1,0},\ldots, b_{k-1,0},\ldots, b_{k-1,n-k-1})}
{p(b_{0,0},b_{0,1},\ldots,b_{0,n-k-1},b_{1,0},\ldots, b_{k-1,0},\ldots, b_{k-1,n-k-1})},
\end{equation}
where $p^{(i,j)},p\in \fq[x_1,\ldots,x_{k(n-k)}]$ are $k(n-k)$-variate polynomials of degree at most 1 in each variable and they have no zeros in $\Omega$.
\end{lem}

\begin{proof}        
Recall that a generator matrix $G_{TGRS}=[g_{i,j}]$ ($0\leq i \leq k-1$, $0\leq j \leq n-1$) of $\C(\mL,\mP,B)$ can be given by \eqref{S8-GM}, where 
 \begin{equation}\label{S20}
 g_{i,j}=v_{j+1}(a_{j+1}^{i}+\sum_{s=0}^{n-k-1}b_{i,s}a_{j+1}^{k+s}).
 \end{equation}
Observe that $g_{i,j}$ is the evaluation at $(b_{i,0},\ldots,b_{i,n-k-1})$ of the corresponding polynomial $v_{j+1}(a_{j+1}^{i}+\sum_{s=0}^{n-k-1}x_{(i+1)(s+1)}a_{j+1}^{k+s})$ in $\fq[x_1\ldots x_{k\times{(n-k)}}]$ of degree at most $1$ in each variable. Here $b_{i,s}$ corresponds to the variable $x_{(i+1)(s+1)}$.

Let $G_{TGRS}=[Q^{(B)} \mid T^{(B)}]$, which defines $Q^{(B)} \in\fq^{k\times{k}}$ and $T^{(B)} \in\fq^{k\times{(n-k)}}$. Since $\C(\mL,\mP,B)$ is MDS by the assumption, $Q^{(B)}$ is invertible. Then we have
\begin{eqnarray}\label{S21}
 M^{(B)}={Q^{(B)}}^{-1}T^{(B)}=\frac{{\rm adj}{(Q^{(B)})}T^{(B)}}{{\rm det}(Q^{(B)})},
\end{eqnarray}
where ${\rm adj}(Q^{(B)})$ is the adjugate matrix of $Q^{(B)}$ and ${\rm det}(Q^{(B)})$  is the determinant of $Q^{(B)}$.

The determinant ${\rm det}(Q^{(B)})$ is the evaluation at $B$ of a polynomial $p\in{\mathbb{F}_{q}[x_1\cdots x_{k\times{(n-k)}}]}$, where $p$ can be determined by $Q^{(B)}$. Note that each $b_{i,s}$ appears only in one row of $G_{TGRS}$. Then it is clear that $p$ is of degree at most $1$ in each variable. Thus $p$ has no zeros in $\Omega$ since $Q^{(B)}$ is invertible. This gives the polynomial $p$.

The $(i',j')$$\mbox{-}$th entry of the matrix ${\rm adj}{(Q^{(B)})}T^{(B)}$ is equal to the inner product of the $i'$-th row of ${\rm adj}{(Q^{(B)})}$ and the $j'$-th column of $T^{(B)}$. Then the $i'\times j'$ entry of $ M^{(B)}$ can be expressed by $b_{i,s}$'s. By associating $b_{i,s}$ with the variable $x_{(i+1)(s+1)}$, we obtain the polynomial $p^{(i',j')}\in \fq[x_1\cdots x_{k{(n-k)}}]$. By the definition of adjugate matrix, it can be verified that $p^{(i',j')}$ is of degree at most $1$ in each variable. Furthermore, $p^{(i',j')}$ has no zeros in $\Omega$, otherwise $G^{(sys,B)}$ contains a row with $k$ zeros, contradicting the assumption that $\C(\mL,\mP,B)$ is MDS. This completes the proof.
 \end{proof}

\begin{remark}
Note that in Lemma \ref{L9} the polynomials $p^{(i,j)}$ and $p$ in $\fq[x_1,\ldots,x_{k(n-k)}]$ can be explicitly computed for given $\alpha$ and $\nu$, and their coefficients do not depend on the coefficient matrix $B$.
\end{remark}

\begin{thm}\label{T6}
Let $\C(\mL,\mP,B)$ be an $[n,k]$ code defined by \eqref{S6} and $\Omega$ be given as in \eqref{Omega}. Assume that $\min\{k,n-k\}\geq{3}$ and there is a $\widetilde{B} \in \Omega$ for $\C(\mL,\mP,\widetilde{B})$ to be a non-GRS MDS code. Then there is a non-zero multivariate polynomial $P\in \fq[x_1,\ldots,x_{k(n-k)}]$ with degree at most 6 in each variable such that all $B\in \Omega$ for which $\C(\mL,\mP,B)$ is GRS are zeros of $P$.
\end{thm}

 \begin{proof}
Note that $B\in \Omega$, namely, $\C(\mL,\mP,B)$ is MDS. Let $G^{(sys,B)}=[I_k|M^{(B)}]$ be the systematic generator matrix of $\C(\mL,\mP,B)$, and $M'^{(B)}=[M'^{(B)}_{i,j}]\in \fq^{k\times(n-k)} $ with $M'^{(B)}_{i,j}=(M^{(B)}_{i,j})^{-1}$. By Lemma \ref{L8}, the MDS code $\C(\mL,\mP,B)$ is a GRS code if and only if all $3\times3$ minors of $M'^{(B)}$ are zero.
Assume that there is a $\widetilde{B} \in \Omega$ such that $\C(\mL,\mP,\widetilde{B})$ is a non-GRS code.
Then there is at least one nonzero $3\times{3}$ minor of $M'^{(\widetilde{B})}$. Fix this minor for all $B\in \Omega$. We focus on this $3\times{3}$ minor at the same position. 

By Lemma \ref{L9}, the entry $M'^{(B)}_{i,j}$ of the matrix $M'^{(B)}$ is the evaluation at $B$ of the polynomial $p/p^{(i,j)}\in{\fq [x_1,\ldots,x_{k(n-k)}]}$. Then the $3\times{3}$ minor of $M'^{(B)}$ associated with the fixed $3\times{3}$ minor of $M'^{(\widetilde{B})}$ can be expressed as the evaluation at $B$ of the polynomial $p^3 P/Q$,  where $P$ and $Q$ are given by
\begin{itemize}
\item $P$ is the sum of products of any six $p^{(i,j)}$'s associated with the fixed $3\times{3}$ minor;
\item $Q$ is the product of all nine $p^{(i,j)}$'s associated with the fixed $3\times{3}$ minor.
\end{itemize}
Note that $p^{(i,j)}$ and $p$ in $\fq[x_1,\ldots,x_{k(n-k)}]$ are $k(n-k)$-variate polynomials of degree at most 1 in each variable and they have no zeros in $\Omega$.

It follows that $P$  is a polynomial of degree at most $6$ in each variable. Then the fixed $3\times 3$ minor of $M'^{(B)}$ is equal to $0$ if and only if the evaluation at $B$ of the polynomial $P$ are zero, due to the fact that the polynomials $Q$ and $p$ have no zeros in $\Omega$ by Lemma \ref{L9}. Since the evaluation at $\widetilde{B}$ of $P$ is nonzero, it implies that $P$ is a nonzero polynomial. It is clear that the evaluation at $B$ of the polynomial $P$ is zero if $\C(\mL,\mP,B)$ is GRS for such $B$. This completes the proof.
 \end{proof}

\begin{remark}
Note that in Theorem \ref{T6} the polynomial $P\in \fq[x_1,\ldots,x_{k(n-k)}]$ can be explicitly computed for given $\alpha$, $\nu$ and $\widetilde{B} \in \Omega$, where $\C(\mL,\mP,\widetilde{B})$ is a non-GRS MDS code. Note that the polynomial $P$ may be not unique by the proof of Theorem \ref{T6}.
\end{remark}

Theorem \ref{T6} can be interpreted as follows: for any given $n$, $k$, $\boldsymbol{\alpha}$ and $\boldsymbol{\nu}$, either all MDS codes are GRS codes, or the number of GRS codes are upper bounded by the number of zeros of a nonzero multi-variable polynomial $P$ of degree at most $6$ in each variable.

\begin{example}
Let $n=6,\ k=3,\ q=17$, $\alpha=(1,2,3,4,5,6)\in{\mathbb{F}_{17}^6}$, $\boldsymbol{\nu}=(1,\ldots,1)$, and the coefficient matrix $B$ be of the form
\begin{equation*}
B=\begin{pmatrix}
x&0&0\\
0&0&0\\
0&0&y
\end{pmatrix},
\end{equation*}
where $x,y\in \fq$.
According to the proof of Lemma \ref{L9}, with some computation using Magma, the polynomials $p$ and $p^{(i,j)}$'s with two variables $x$ and $y$ are given by
\begin{equation*}
p=11xy+12x+10y+2,
\end{equation*} 
\begin{equation*}
\begin{pmatrix}
 p^{(0,0)}&p^{(0,1)}&p^{(0,2)}\\
 p^{(1,0)}&p^{(1,1)}&p^{(1,2)}\\
 p^{(2,0)}&p^{(2,1)}&p^{(2,2)}
\end{pmatrix}
=
\begin{pmatrix}
 xy+14x+9y+2&5xy+10x+12y+6&13xy+7x+15y+12\\
 16xy+13x+6y+11&10xy+15x+7y+1&12xy+4x+8y+4\\
 15xy+14x+12y+6&10xy+x+8y+12&4xy+2x+4y+3
\end{pmatrix}.
\end{equation*}
Note that $\C(\mL,\mP,B)$ is a non-GRS MDS code if $(x,y)=(9,9)$ by Magma. According to the proof of Theorem \ref{T6}, with some computation using Magma, the polynomial $P$ is given by 
\begin{align*}
P(x,y)&=7x^6y^5 + 15x^6y^4 + x^6y^3 + 2x^6y + 3x^5y^6 + 5x^5y^5 + 6x^5y^4 +12x^5y^3 + 11x^5y^2 + 2x^5y + 3x^5 + 13x^4y^6 \\
&+ 14x^4y^5+16x^4y^4 + 4x^4y^3 + 12x^4y^2 + 15x^4y + x^4 + 10x^3y^6 +9x^3y^5 + 8x^3y^4 + 5x^3y^3 + 5x^3y^2+ 15x^3y\\
& + x^3 + 12x^2y^6+ 12x^2y^5 + 7x^2y^4 + 3x^2y^2 + 16x^2y + 14x^2 + 5xy^6 +16xy^5 + 6xy^4 + 11xy^3 + 15xy^2\\
& + 3xy + 4x + 8y^6 + 2y^5 +13y^4+ 16y^3 + 13y^2 + 4y.
\end{align*}
%Define the set $R$ as the roots of $P\in\mathbb{F}_{17}$, which is given as
%\begin{align*}
%R&=\{ <15, 4>, <3, 9>, <14, 0>, <9, 8>, <16, 10>, <14, 11>, <8, 16>, <12, 4>, <6,5>, <11, 3>\\
% &,<14, 6>, <8, 5>, <4, 10>,<13, 8>, <5, 15>, <16, 7>, <0, 10>, <12,7>, <4, 1>, <5, 10>, <9, 11>\\
% &,<4, 6>, <9, 5>, <2, 6>, <6, 12>, <2, 5>, <7, 8>,<2, 13>, <10, 14>, <4, 11>, <15, 1>, <16, 9> \\
% &,<9, 6>, <14, 16>, <14, 3>, <11,8>, <5, 2>, <9, 1>, <10, 13>, <1, 5>, <11, 16>,<7, 7>, <0, 0>\\
% &,<7, 12>, <12, 8>\}
%\end{align*}
Furthermore, Magma experiments show that the number of zeros of $P$ is $45$, and when $x$ and $y$ run through $\fq$, the number of MDS codes is $90$, the number of GRS codes is $8$ and the number of non-GRS MDS codes is $82$.
\end{example}

\section{Conclusions} \label{sect-6}
In this paper, we take an in-depth study on the $(\mL,\mP)$-TGRS codes for the most general case. Our main contributions are summarized as follows:
\begin{itemize}
\item  We presented a concise necessary and sufficient condition for $(\mL,\mP)$-TGRS codes to be MDS by a universal method, which extends related results in the literature. Additionally,  we proposed a sufficient condition for $(\mL,\mP)$-TGRS codes to be NMDS under the condition that the code is self-dual.
\item  We explicitly characterized the parity check matrices of  $(\mL,\mP)$-TGRS codes and presented a sufficient condition such that the $(\mL,\mP)$-TGRS codes are self-dual.
\item We investigated the non-GRS properties of $(\mL,\mP)$-TGRS codes by using Schur squares and combinatorial techniques. As a result, a large infinite family of non-GRS MDS codes was obtained.
\end{itemize}

The following interesting problems naturally arise:
\begin{problem}
Characterize the necessary and sufficient condition such that the $(\mL,\mP)$-TGRS codes defined by \eqref{S6} is NMDS for the general case.
\end{problem}
\begin{problem}
Construct explicit new infinite families of non-GRS MDS codes, NMDS codes, $m$-MDS codes, self-dual codes from the $(\mL,\mP)$-TGRS codes.
\end{problem}

The reader is cordially invited to join the adventure and solve the problems above.

\section*{Acknowledgements}
This work was supported by the National Natural Science Foundation of China (Nos. 12471492, 12401688), the Innovation Group Project of the Natural Science Foundation of Hubei Province of China (No. 2023AFA021) and the Natural Science Foundation of Hubei Province of China (No. 2024AFB419).


\begin{thebibliography}{99}% more than 9 --> 99 / less than 10 --> 9
\bibitem{BART} D. Bartoli, M. Giulietti, I. Platoni, On the covering radius of MDS codes. IEEE Trans. Inf. Theory 61(2): 801-811 (2015).
\bibitem{BEBO} P. Beelen, M. Bossert, S. Puchinger, J. Rosenkilde, Structural properties of twisted Reed-Solomon codes with applications to cryptography. 2018 IEEE Int. Symp. Inf. Theory (ISIT), Vail, CO, USA, 2018, pp. 946-950.
\bibitem{BEPU} P. Beelen, S. Puchinger, J.R. n{\'e} Nielsen, Twisted Reed-Solomon codes. 2017 IEEE Int. Symp. Inf. Theory (ISIT), Aachen, Germany, 2017, pp. 336-340.
\bibitem{BERO} P. Beelen, S. Puchinger, J. Rosenkilde, Twisted Reed-Solomon codes. IEEE Trans. Inf. Theory 68(5): 3047-3061 (2022).
\bibitem{BETS} K. Betsumiya, S. Georgiou, T.A. Gulliver, M. Harada, C. Koukouvinos, On self-dual codes over some prime fields. Discret. Math. 262(1-3): 37-58 (2003).
\bibitem{AMAR} M.A. Boer, Almost MDS Codes. Des. Codes Cryptogr. 9(2): 143-155 (1996).
\bibitem{CASC} I. Cascudo, R. Cramer, D. Mirandola, G. Z{\'e}mor, Squares of random linear codes. IEEE Trans. Inf. Theory 61(3): 1159-1173 (2015).
\bibitem{CHEN} W. Cheng, On parity-check matrices of twisted generalized Reed-Solomon codes. IEEE Trans. Inf. Theory 70(5): 3213-3225 (2024).
\bibitem{COUV} A. Couvreur, P. Gaborit, V. Gauthier-Uma{\~n}a, A. Otmani, J.P. Tillich, Distinguisher-based attacks on public-key cryptosystems using Reed-Solomon codes. Des. Codes Cryptogr. 73(2): 641-666 (2014).
\bibitem{CRAM} R. Cramer, V. Daza, I. Gracia, J.J. Urroz, G. Leander, J. Mart{\'\i}-Farr{\'e}, C. Padr{\'o}, On codes, matroids, and secure multiparty computation from linear secret-sharing schemes. IEEE Trans. Inf. Theory 54(6): 2644-2657 (2008).
\bibitem{CRAD} R. Cramer, I. Damgard, J.B. Nielsen, Secure multiparty computation and secret sharing. Cambridge University Press, 2015.
\bibitem{DING} Y. Ding, S. Zhu, New self-dual codes from TGRS codes with general $\ell$ twists. Advances in Mathematics of Communications 19(2): 662-675 (2025).
\bibitem{DODU} S.M. Dodunekov, I.N. Landjev, Near-MDS codes over some small fields. Discret. Math. 213(1-3): 55-65 (2000).
\bibitem{DOUG} S.T. Dougherty, S. Mesnager, P. Sol{\'e}, Secret-sharing schemes based on self-dual codes. 2008 IEEE Information Theory Workshop, Porto, Portugal, 2008, pp. 338-342.
\bibitem{FORN} G. Forney, Generalized minimum distance decoding. IEEE Trans. Inf. Theory 12(2): 125-131 (1966).
\bibitem{GU} H. Gu, J. Zhang, On twisted generalized Reed-Solomon codes with $\ell$ twists. IEEE Trans. Inf. Theory 70(1): 145-153 (2024).
\bibitem{GULL} T.A. Gulliver, J.K. Kim, Y. Lee, New MDS or near-MDS self-dual codes. IEEE Trans. Inf. Theory 54(9): 4354-4360 (2008).
\bibitem{HUAQ} D. Huang, Q. Yue, Y. Niu, MDS or NMDS LCD codes from twisted Reed-Solomon codes. Cryptogr. Commun. 15(2): 221-237 (2023).
\bibitem{HUAN} D. Huang, Q. Yue, Y. Niu, X. Li, MDS or NMDS self-dual codes from twisted generalized Reed-Solomon codes. Des. Codes Cryptogr. 89(9): 2195-2209 (2021).
\bibitem{HWPV} W.C. Huffman, V. Pless, Fundamentals of Error-Correcting Codes. Cambridge University Press, Cambridge, 2003.
\bibitem{LAVA} J. Lavauzelle, J. Renner, Cryptanalysis of a system based on twisted Reed-Solomon codes. Des. Codes Cryptogr. 88(7): 1285-1300 (2020).
\bibitem{LIUH} H. Liu, S. Liu, Construction of MDS twisted Reed-Solomon codes and LCD MDS codes. Des. Codes Cryptogr. 89(9): 2051-2065 (2021).
%\bibitem{MACW} F.J. MacWilliams, N.J.A. Sloane, The theory of error correcting codes. Amsterdam, The Netherlands: Elsevier, 1977.
\bibitem{MIRA} D. Mirandola, Schur products of linear codes: a study of parameters. Master Thesis (under the supervision of G. Z{\'e}mor), Univ. Bordeaux 1 and Stellenbosch Univ., July 2012.  Available:
http://www.algant.eu/documents/theses/mirandola.pdf
%\bibitem{NGUY} J.P. Nguyen, Applications of Reed-Solomon codes on optical media storage. San Diego State University, 2011.
\bibitem{RAND} H. Randriambololona, On products and powers of linear codes under componentwise multiplication. Algorithmic arithmetic, geometry, and coding theory 637: 3-78 (2015).
\bibitem{ROTL} R.M. Roth, A. Lempel, On MDS codes via Cauchy matrices. IEEE Trans. Inf. Theory 35(6): 1314-1319 (1989).
\bibitem{ROTS} R.M. Roth, G. Seroussi, On generator matrices of MDS codes. IEEE Trans. Inf. Theory 31 (6): 826-830 (1985).
\bibitem{SINC} H. Singh, K.C. Meena, MDS multi-twisted Reed-Solomon codes with small dimensional hull. Cryptogr. Commun. 16(3): 557-578 (2024).
\bibitem{SING} R.C. Singleton, Maximum distance $q$-nary codes. IEEE Trans. Inf. Theory 10(2): 116-118 (1964).
%\bibitem{SOBC} G. Sobczyk, Generalized Vandermonde determinants and applications. Aportaciones Matematicas, Serie Comunicaciones. 30: 203-213 (2002).
\bibitem{STEA} A.M. Steane, Error correcting codes in quantum theory. Phys. Rev. Lett. 77(5): 793 (1996).
\bibitem{SUIL} J. Sui, Q. Yue, X. Li, D. Huang, MDS, near-MDS or 2-MDS self-dual codes via twisted generalized Reed-Solomon codes. IEEE Trans. Inf. Theory 68(12): 7832-7841 (2022).
\bibitem{SUIS} J. Sui, Q. Yue, F. Sun, New constructions of self-dual codes via twisted generalized Reed-Solomon codes. Cryptogr. Commun. 15(5): 959-978 (2023).
\bibitem{SUIZ} J. Sui, X. Zhu, X. Shi, MDS and near-MDS codes via twisted Reed-Solomon codes. Des. Codes Cryptogr. 90(8): 1937-1958 (2022).
%\bibitem{WICK} S.B. Wicker, V.K. Bhargava, Reed-Solomon codes and their applications. Hoboken, NJ, USA: Wiley, 1999.
\bibitem{ZHANGA} A. Zhang, K. Feng, On the constructions of MDS self-dual codes via cyclotomy. Finite Fields Their Appl. 77: 101947 (2022).
\bibitem{ZHANGJ} J. Zhang, Z. Zhou, C. Tang, A class of twisted generalized Reed-Solomon codes. Des. Codes Cryptogr. 90(7): 1649-1658 (2022).
\bibitem{ZHAO}C. Zhao, W. Ma, T. Yan, Y. Sun, Research on the construction of maximum distance separable codes via arbitrary twisted generalized Reed-Solomon codes. ArXiv:2408.12049 (2024).
\end{thebibliography}
\end{document}